\documentclass[11pt,letter]{article}

\usepackage{amsmath}
\usepackage{amsthm}
\usepackage{amssymb}
\usepackage{verbatim}
\usepackage{fullpage}
\usepackage[margin=.9in]{geometry}
\usepackage{framed}
\usepackage{algorithm,algpseudocode}
\usepackage{caption}
\usepackage{subcaption}
\usepackage{subfig}
\usepackage{graphicx}
\usepackage{complexity}
\newcommand{\bdisp}{\begin{displaystyle}}
\newcommand{\edisp}{\end{displaystyle}}

\renewcommand{\Pr}{\mathbb{P}}

\newcommand{\Exp}{\mathbb{E}}
\newcommand{\from}{\leftarrow}

\renewcommand{\poly}{\mathrm{poly}}
\renewcommand{\polylog}{\mathrm{polylog}}

\newcommand{\Normal}{\mathcal{N}}
\newcommand{\Unif}{\mathrm{Unif}}

\renewcommand{\vec}[1]{\mathbf{#1}}

\newcommand{\Real}{\mathbb{R}}
\newcommand{\Rational}{\mathbb{Q}}

\newcommand{\argmin}{\operatorname*{\mathrm{arg\,min}}}

\renewcommand{\d}{\mathrm{d}}
\newcommand{\Diff}[2][]{\frac{\d#1}{\d#2}}
\newcommand{\Grad}{\nabla}
\newcommand{\Del}[2][]{\frac{\partial#1}{\partial#2}}

\newcommand{\ferr}{\epsilon}
\newcommand{\perr}{\delta}
\newcommand{\safety}{s}

\newcounter{nTheorems}
\numberwithin{nTheorems}{section}

\newtheorem{theorem}[nTheorems]{Theorem}
\newtheorem{corollary}[nTheorems]{Corollary}

\newtheorem{lemma}[nTheorems]{Lemma}
\newtheorem{proposition}[nTheorems]{Proposition}
\newtheorem*{proposition*}{Proposition}
\newtheorem*{lemma*}{Lemma}

\newtheorem{definition}[nTheorems]{Definition}
\newtheorem{example}[nTheorems]{Example}

\newtheoremstyle{break}
  {\topsep}{\topsep}%
  {}{}%
  {\bfseries}{}%
  {\newline}{}%
\theoremstyle{break}
\newtheorem{myalgorithm}{Algorithm}

\title{Optimizing Star-Convex Functions}
\author{
Jasper C.H. Lee
\quad
Paul Valiant\\ \ \\
Department of Computer Science\\
Brown University\\
\texttt{\{jasperchlee,paul\char`_valiant\}@brown.edu}
}

\begin{document}
\maketitle
\begin{abstract}
We introduce a polynomial time algorithm for optimizing the class of star-convex functions, under no Lipschitz or other smoothness assumptions whatsoever, and no restrictions except exponential boundedness on a region about the origin, and Lebesgue measurability. The algorithm's performance is polynomial in the requested number of digits of accuracy and the dimension of the search domain. This contrasts with the previous best known algorithm of Nesterov and Polyak which has exponential dependence on the number of digits of accuracy, but only $n^\omega$ dependence on the dimension $n$ (where $\omega$ is the matrix multiplication exponent), and which further requires Lipschitz second differentiability of the function~\cite{Nesterov:2006}. Star-convex functions constitute a rich class of functions generalizing convex functions, including, for example: for any convex (or star-convex) functions $f$, $g$, with global minima $f(\vec{0})=g(\vec{0})=0$, their \emph{power mean} $h(x)=\left(\frac{f(x)^p+g(x)^p}{2}\right)^{1/p}$ is star-convex, for \emph{any} real $p$, defining powers via limits as appropriate. Star-convex functions arise as loss functions in non-convex machine learning contexts, including, for data points $X_i$, parameter vector $\theta$, and any real exponent $p$, the loss function $h_{\theta,X}(\hat{\theta})=\left(\sum_i |(\hat{\theta}-\theta)\cdot X_i|^p\right)^{1/p}$, significantly generalizing the well-studied \emph{convex} case where $p\geq 1$. Further, for \emph{any} function $g>0$ on the surface of the unit sphere (including discontinuous, or pathological functions that have different behavior at rational vs. irrational angles), the star-convex function $h(x)=||x||_2\cdot g\left(\frac{x}{||x||_2}\right)$ extends the arbitrary behavior of $g$ to the whole space.

Despite a long history of successful gradient-based optimization algorithms, star-convex optimization is a uniquely challenging regime because 1) gradients and/or subgradients often do not exist; and 2) even in cases when gradients exist, there are star-convex functions for which gradients provably provide no information about the location of the global optimum. We thus bypass the usual approach of relying on gradient oracles and introduce a new randomized cutting plane algorithm that relies only on function evaluations. Our algorithm essentially looks for structure at all scales, since, unlike with convex functions, star-convex functions do not necessarily display simpler behavior on smaller length scales. Thus, while our cutting plane algorithm refines a feasible region of exponentially decreasing volume by iteratively removing ``cuts", unlike for the standard convex case, the structure to efficiently discover such cuts may not be found within the feasible region: our novel star-convex cutting plane approach discovers cuts by sampling the function exponentially far outside the feasible region.

We emphasize that the class of star-convex functions we consider is as unrestricted as possible: the class of Lebesgue measurable star-convex functions has theoretical appeal, introducing to the domain of polynomial-time algorithms a huge class with many interesting pathologies.
We view our results as a step forward in understanding the scope of optimization techniques beyond the garden of convex optimization and local gradient-based methods.

\end{abstract}
\thispagestyle{empty}\setcounter{page}{0}
\newpage
\section{Introduction}

Optimization is one of the most influential ideas in computer science, central to many rapidly developing areas within computer science, and also one of the primary exports to other fields, including operations research, economics and finance, bioinformatics, and many design problems in engineering. Convex optimization, in particular, has produced many general and robust algorithmic frameworks that have each become fundamental tools in many different areas: linear programming has become a general modeling tool, its simple structure powering many algorithms and reductions; semidefinite programming is an area whose scope is rapidly expanding, yielding many of the best known approximation algorithms for optimizing constraint satisfaction problems~\cite{chlamtac-tulsiani}; convex optimization generalizes both of these and has introduced powerful optimization techniques including interior point and cutting plane methods. Our developing understanding of optimization has also led to new algorithmic design principles, which in turn leads to new insights into optimization. Recent progress in algorithmic graph theory has benefited enormously from the optimization perspective, as many recent results on max flow/min cut~\cite{Madry:2013,Kelner:2014,Lee:2014,Christiano:2011,Sherman:2013}, bipartite matching~\cite{Madry:2013}, and Laplacian solvers~\cite{Vishnoi:2012,Christiano:2011,Spielman:2004} have made breakthroughs that stem from developing a deeper understanding of the characteristics of convex optimization techniques in the context of graph theory. These successes motivate the quest for a deeper and broader understanding of optimization techniques: 1) to what degree can convex optimization techniques be extended to non-convex functions; 2) can we develop general new tools for tackling non-convex optimization problems; and 3) can new techniques from non-convex optimization yield new insights into convex optimization?

As a partial answer to the first question, gradient descent---perhaps the most natural optimization approach---has had enormous success recently in a variety of practically-motivated non-convex settings, sometimes with provable guarantees. The method and its variants are the de facto standard for training (deep) neural networks, a hot topic in high dimensional non-convex optimization, with many recent practical results (e.g.~\cite{Imagenet,AlphaGo}). Gradient descent can be thought of as a ``greedy" algorithm, which repeatedly chooses the most attractive direction from the local landscape. The efficacy of gradient descent algorithms relies on local assumptions about the function: for example, if the first derivative is Lipschitz (slowly varying), then one can take large downhill steps in the gradient direction without worrying that the function will change to going uphill along this direction. Thus the convergence of gradient descent algorithms typically depends on a Lipschitz parameter (or other smoothness measure), and conveniently does not depend on the dimension of the search space. Many of these algorithms converge to within $\ferr$ of a local optimum in time $poly(1/\ferr)$, with additional polynomial dependence on the Lipschitz constant or other appropriate smoothness guarantee~\cite{Nesterov:AGD,Hazan:2015}. In cases where all local optima are global optima, then one has global convergence guarantees.

While one intuitively expects gradient descent algorithms to always converge to a local minimum, in this paper we study the optimization of a natural generalization of convex functions where, despite the global optimum being the only stationary point/local minimum for every function in this class, provably no variant of gradient descent converges in polynomial time. The class of star-convex functions, which we define below, includes many functions of both practical and theoretical interest, both generalizing common families of convex functions to wider parameter regimes, and introducing new ``pathologies" not found in the convex case. We show, essentially, how to make a gradient-based cutting plane algorithm ``robust" to many new pathologies, including lack of gradients or subgradients, long narrow ridges and rapid oscillation in directions transverse to the global minimum, and, in fact, almost arbitrary discontinuities. This challenging setting gives new insights into what fundamentally enables cutting plane algorithms to work, which we view as progress towards answering questions 2 and 3 above.

\subsection{Star-convex functions}
This paper focuses on the optimization of \emph{star-convex} functions, a particular class of (typically) non-convex functions that includes convex functions as a special case.
We define these functions as follows, based on the definition in Nesterov and Polyak~\cite{Nesterov:2006}.

\begin{definition}[Star-convex functions]
A function $f : \Real^n \to \Real$ is \emph{star-convex} if there is a global minimum $x^\ast \in \Real^n$ such that for all $\alpha \in [0,1]$ and $x \in \Real^n$,
$$ f(\alpha x^\ast + (1-\alpha)x) \le \alpha f(x^\ast) + (1-\alpha) f(x) $$
We call $x^\ast$ the star center of $f$.
For convenience, in the rest of the paper we shall refer to the minimum function value as $f^\ast$.
\end{definition}

See Appendix~\ref{sec:examples} for a variety of basic constructions and examples of star-convex functions.

Intuitively, if we visualize the objective function as a landscape, star-convexity means that the global optimum is ``visible" from every point---there are no ``ridges" on the way to the global optimum, but there could be many ridges in transverse directions. Since the global optimum is always visible in a downhill direction from every point, gradient descent methods would \emph{seem} to be effective.
Counterintuitively, these methods all fail in general. In Section~\ref{sect:Hardness} we demonstrate that standard gradient descent and cutting plane methods fail in the absence of Lipschitz guarantees.

\begin{figure}
\centering
\begin{subfigure}[b]{.45\textwidth}
\centering
\includegraphics[width=5cm, keepaspectratio]{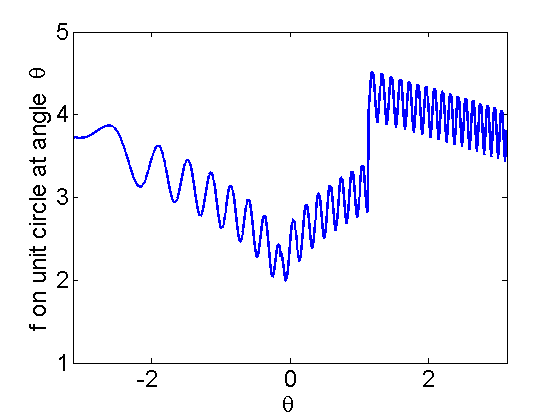}
\caption{$g$ defined on the unit circle}
\label{Fig:fig1a}
\end{subfigure}
\quad
\begin{subfigure}[b]{.45\textwidth}
\centering
\includegraphics[width=5cm, keepaspectratio]{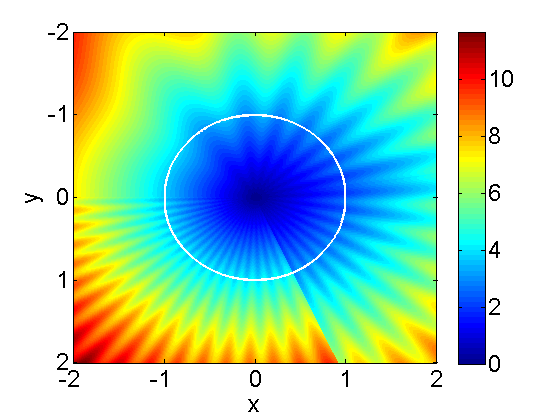}
\caption{Linear extension of $g$ to $f$}
\label{Fig:fig1b}
\end{subfigure}
\caption{An example star-convex function $f$ defined by linearly extrapolating an \emph{arbitrary} positive function $g$ defined on the unit circle (in white).}
\label{Fig:fig1}
\end{figure}

One broad class of star-convex functions that gives a good sense of the scope of this definition is constructed by the following process: 1) pick an arbitrary positive function $g(\theta)$ on the unit circle (see Figure~\ref{Fig:fig1a}) which may be discontinuous, rapidly oscillating, or otherwise badly behaved; and 2) linearly extend this function to a function $f$ on the entire plane, about the value $f(0)=0$ (see Figure~\ref{Fig:fig1b}).

The name ``star-convex'' comes from the fact that each sublevel set (the set of $x$ for which $f(x)<c$, for some $c$) is ``star-shaped".

At any point at which a gradient or subgradient exists, the star center (global optimum) lies in the halfspace opposite the (sub)gradient. However, even for the simple example of Figure~\ref{Fig:fig1}, gradients do not exist for angles $\theta$ at which $g$ is discontinuous, and subgradients do not exist at most angles, including those angles where $g$ is a local maxima with respect to $\theta$. Further, even for differentiable star-convex functions, gradients can be misleading, since a rapidly oscillating $g$ implies that gradients typically point in the transverse direction, nearly orthogonal to the direction of the star center.

While it is often standard to design algorithms assuming one has access to an oracle that returns both function values and gradients, we instead only assume access to the function value: even in the case when gradients exist, it is unclear whether they are algorithmically helpful in our setting. (See Section~\ref{sect:Hardness} for details.) Indeed, we pose this as an open problem: under what assumptions (short of the Lipschitz guarantees on the second derivative of Nesterov and Polyak~\cite{Nesterov:2006}), can one meaningfully use a gradient oracle to optimize star-convex functions?

In this paper, we show that, assuming only Lebesgue measurability and an exponential bound on the function value within a large ball, our algorithm optimizes a star-convex function in $\polylog(1/\ferr)$ time where $\ferr$ is the desired accuracy in function value.

\begin{theorem}
\label{thm:InformalIntro} (Informal)
Given evaluation oracle access to a Lebesgue measurable star-convex function $f : \Real^n \to \Real$ and an error parameter $\ferr$, with the guarantee that $x^\ast$ is within radius $R$ of the origin, Algorithm~\ref{alg:Ellipsoid} returns an estimate $f_0$ of the minimum function value such that $|f_0 - f^\ast| \le \ferr$ with high probability in time $\poly(n, \log \frac{1}{\ferr}, \log R)$.
\end{theorem}

In previous work, Nesterov and Polyak~\cite{Nesterov:2006} introduced a new adaptation of Newton's method to find local minima in functions that are twice-differentiable, with Lipschitz continuity of their second derivatives. For the class of Lipschitz-twice-differentiable star-convex functions, they show that their algorithm converges to within $\ferr$ of the global optimum in time $O(\frac{1}{\sqrt{\ferr}})$, using star convexity to lower bound the amount of progress in each of their optimization steps.

Our results are stronger in two significant senses: 1) as opposed to assuming Lipschitz twice-differentiability, we make \emph{no continuity assumptions whatsoever}, optimizing over an essentially arbitrary measurable function within the class of star-convex functions; 2) while the algorithm of~\cite{Nesterov:2006} requires exponential time to estimate the optimum to $k$ digits of accuracy, our algorithm is \emph{polynomial in the requested number of digits of accuracy}.

The reader may note that the complexity of our algorithm does depend polynomially on the number of dimensions $n$ of the search space, whereas the Nesterov-Polyak algorithm uses a number of calls to a Hessian oracle that is \emph{independent} of the dimension $n$.
However, our dependency on the dimension $n$ is necessary in the absence of Lipschitz guarantees, even for convex optimization~\cite{Nemirovskii:1983}.

We further point out an important distinction between \emph{star-convex} optimization and \emph{star-shaped} optimization.
A star-shaped function is defined similarly to a star-convex function, except the star center is allowed to be located away from the global minimum.
Certain \NP-hard optimization problems, including max-clique can be rephrased in terms of star-shaped optimization~\cite{MaxClique}, and the problem is in general impossible without continuity guarantees, for the global minimum may be hidden along a single ray from the star center that is discontinuous from the rest of the function. In this sense, our star-convex optimization algorithm narrowly avoids solving an \NP-hard optimization problem.

\subsection{Applications of Star-Convex Functions}

We begin with the following concrete example: the function $(|x|^p+|y|^p)^{1/p}$ is convex for $p\geq 1$, yet is star-convex for \emph{all} real $p$, positive or negative, taking limits as appropriate for $p=0$. We generalize this example significantly below.

\emph{Empirical risk minimization} is a central technique in machine learning~\cite{Vapnik:2013}, where, given training data $x_i$ with labels $y_i$, we seek a hypothesis $h$ so as to minimize \[\frac{1}{m}\sum_{i=1}^m L(h(x_i),y_i),\] where $L$ is a \emph{loss function} that describes the penalty for misprediction, on input our prediction $h(x_i)$ and the true answer $y_i$. We take the hypothesis $h$ to be a linear function (this setting includes kernel methods that preprocess $x_i$ first and then apply a linear function). If the loss function $L$ is convex, then finding the optimal hypothesis $h$ is a convex optimization problem, which is widely used in practice: for exponent $p\geq 1$, we let $L(\hat{y}_i,y_i)=|\hat{y}_i-y_i|^p$ and thus aim to optimize the convex function \begin{equation}\label{eq:ERM-p}\argmin_{h}\sum_{i=1}^m |h\cdot x_i - y_i|^p.\end{equation}

In this paper we draw attention to the interesting regime of Equation~\ref{eq:ERM-p} where $p<1$. This regime is not accessible with standard techniques. It is straightforward to verify, however, that the $(1/p)$th power of the objective function of Equation~\ref{eq:ERM-p} is star-convex, and thus the main result of our paper yields an efficient optimization algorithm. Slightly more generally:

\begin{corollary}
  For a loss function $L(\hat{y}_i,y_i)$ such that there exists a real number $p$ for which $L(\cdot,\cdot)^p$ is a star-convex function of its first argument, the following empirical risk minimization problem can be optimized to accuracy $\epsilon$ in time $poly(n,\log 1/\epsilon)$ by inputting its $(1/p)$th power to the main algorithm of this paper:
  \[\argmin_{h}\frac{1}{m}\sum_{i=1}^m L(h(x_i),y_i),\]
\end{corollary}

(We note that negative $p$ is a somewhat peculiar regime in the above corollary, where running a minimization algorithm on the $(1/p)$th power of the loss would actually yield a global \emph{maximum} instead of minimum; we include this case here for completeness' sake.)

The regime where $p\in (0,1)$ has the structure that, when one is far from the right answer, small changes to the hypothesis do not significantly affect performance, but the closer one gets to the true parameters, the richer the landscape becomes. One of many interesting settings with this behavior is biological evolution, where ``fit" creatures have a rich landscape to traverse, while drastically unfit creatures fail to survive. A paper by one of the authors proposed star-convex optimization as a regime where evolutionary algorithms might be unexpectedly successful. This current work finds an affirmative answer to an open question raised there by showing the first polynomial time algorithm for optimizing this general class of functions. (The previous results by Nesterov and Polyak~\cite{Nesterov:2006} fail to apply to this setting because the objective function has regions that behave---for some parameter regimes---like $\sqrt{|x|}$, which is not Lipschitz.)

\begin{corollary}[Extending Theorem 4.5 of~\cite{Valiant:2014}]
There is a single mutation algorithm under which for any $p>0$, defining the loss function $L(\hat{y},y)=|\hat{y}-y|^p$, the class of constant-degree polynomials from $\mathbb{R}^n\rightarrow\mathbb{R}^m$ with bounded coefficients is evolvable with respect to all distributions over the radius $r$ ball.
\end{corollary}


\subsection{Need for Novel Algorithmic Techniques}
\label{sect:Hardness}

Before we outline our approach to optimizing star-convex functions, we demonstrate the need for novel algorithmic approaches by explaining how a wide variety of standard techniques fail on this challenging class of functions. We start by explaining how simple star-convex functions such as $(\sqrt{|x|}+\sqrt{|y|})^2$ confound gradient descent, cutting plane and related approaches, and end with a pathological example with information-theoretic guarantees about its hardness to optimize, even given arbitrarily accurate first-order (gradient) oracle access.

Consider the star-convex function $(\sqrt{|x|}+\sqrt{|y|})^2$, which can be thought of essentially as $\sqrt{|x|}+\sqrt{|y|}$. This function has unique global minimum (and star center) at the origin. Gradients of this function go to infinity as either $x$ or $y$ goes to 0, thus the function essentially has deep canyons along both the $x$ and $y$ axes. Intuitively, this function is challenging to optimize because, despite the origin lying in a ``downhill" direction from every point, gradient descent will typically converge to \emph{one} of the axes first, at which point the gradient has near-infinite magnitude in the direction of the nearest axis, and this direction is thus near-orthogonal to the direction of the origin. In short, many variants of gradient descent will have the following behavior: rapidly jump towards one of the two axes, and then fail to make significant further progress as the search point oscillates around that axis. This is in part a reflection of the fact that the gradient of $(\sqrt{|x|}+\sqrt{|y|})^2$ is \emph{not} Lipschitz, and in fact varies arbitrarily rapidly as $(x,y)$ converges towards either axis; since these axes are ``canyons" in the search space, typical algorithms will spend a disproportionate amount of their time stumbling around this badly-behaved region.

There are many variants of gradient descent constructed to tackle optimization settings that are challenging for ``vanilla" gradient descent. One recent promising line of work concerns \emph{normalized} gradient descent~\cite{Hazan:2015}, which is designed to work well despite regions where the magnitude of the gradient is unusually small or large. However, the analysis techniques rely on the ``strict local quasi-convexity" property, which says that there is a global constant $\epsilon$ such that the downhill halfspace induced by any gradient contains a ball of radius $\epsilon$ about the global optimum; this property does not apply even to the simple function $(\sqrt{|x|}+\sqrt{|y|})^2$, where halfspaces cut arbitrarily close to the origin, the closer one gets to either the $x$ or $y$ axis.

As an alternative to gradient descent methods, one might consider a gradient-based cutting plane algorithm. Cutting plane methods work by iteratively cutting a large search space into ever-smaller regions until they converge to a tiny region guaranteed to contain the global optimum. Cutting plane algorithms, even for convex optimization, are known to typically lead to ill-conditioned behavior, where one dimension of the search region becomes much smaller than the overall diameter of the region, leading to bad numerical properties, difficulty finding and manipulating gradients, etc. These problems become even worse in our star-convex setting, even for the simple function $(\sqrt{|x|}+\sqrt{|y|})^2$: as the search space converges around one of the axes, the gradients become even more strongly oriented towards that axis, yielding cuts that make the already thin region even thinner, instead of cutting in a more productive transverse direction. Indeed, up to numerical precision, the computed gradients may be found to point exactly towards the nearest axis, making cuts in the transverse direction impossible, and thus halting the overall progress of the algorithm indefinitely.

Having seen simple examples showing how standard techniques cannot optimize star-convex functions, we now present a more sophisticated and pathological example which proves that, even with access to a first-order oracle and the assumptions of infinite differentiability and boundedness on a region around the origin, no \emph{deterministic} algorithm can efficiently optimize star-convex functions without Lipschitz guarantees, motivating the somewhat unusual randomized flavor of the algorithm of this paper.
In particular, we show that, for any deterministic polynomial time algorithm, there is an \emph{infinitely differentiable} star-convex function on $\Real^2$ with a unique global minimum near the origin such that for all queried points 1) the returned function value is always 1 and 2) the returned gradient depends \emph{only} on the $y$-coordinate.
The existence of such a function implies the inability of the algorithm to find the $x$-coordinate of the star center.

\begin{example}\label{ex:deterministic}

Given a deterministic polynomial time optimization algorithm $A$, we construct a star-convex function $f_A$ that algorithm $A$ fails to optimize.

We assume that $A$ will only query the function inside some disk of radius $R$ about the origin, for some (possibly exponentially large) $R$.

Now, we first choose a $y$-coordinate $y_0$ (that $A$ cannot specify exactly, for example an irrational number) for the star center; we shall choose $x_0$ later.
Then, we simulate running the algorithm $A$, assuming that in response to any query $(x,y)$ that $A$ makes to the function/gradient oracle, the return value is $f_A(x, y)=1$ and the gradient $\Grad f_A(x, y)$ is 0 in the $x$ direction and $1/(y-y_0)$ in the $y$ direction.

Since $A$ is a polynomial time algorithm, there must only be polynomially many queries before $A$ terminates.
Therefore, we may choose $x_0$ for the star center such that 1) $(x_0, y_0)$ is not exponentially close to collinear with any pair of query points, and 2) the $x$ coordinate of the star center, $x_0$, is chosen to be significantly far from the $x$ coordinates of all query points.

We thus construct a star-convex function $f_A$ subject to the constraints that 1) it passes through all the points simulated above, with corresponding gradients, 2) it has star center at $(x_0,y_0)$, and 3) it is infinitely differentiable, with each derivative being at most exponentially large in the degree of the derivative and the running time of $A$. (Since no pair of query points are close to collinear with the star center, standard interpolation lets us construct $f_A$ by first defining a smooth function $g_A$ on the unit circle and linearly extending $g_A$ to the entire plane, except for an exponentially small region about $(x_0,y_0)$ that we make smooth instead of ``conically" shaped.)

Since the function and gradient oracles return zero information about the $x$-coordinate of the star center, algorithm $A$ clearly has a hopeless task optimizing this function.

\end{example}

In summary, gradient information is very hard to use effectively in star-convex optimization, even for smooth functions. Our algorithm, outlined below, instead only queries the function value---not because we object to gradients, but rather because they do not seem to help.


\subsection{Our Approach}

Our overall approach is via the ellipsoid method, which repeatedly refines an ellipsoidal region containing the star center (global optimum) by iteratively computing ``cuts" that contain the star center, while significantly reducing the overall volume of the ellipsoid.

As mentioned in Section~\ref{sect:Hardness}, even for smooth functions with access to a gradient oracle, the cutting planes induced by the gradients may yield no significant progress in some directions.
Finding ``useful" cutting planes in the star-convex setting requires three novelties---see the introduction to Section~\ref{sec:cuts} for complete details. {\bf First}, the star-convex function may be discontinuous, without even subgradients defined at most points, so we instead rely on a sampling process involving the ``\emph{blurred logarithm}" of the objective function. The blurred logarithm mitigates both the (potentially) exponential range of the function, and the (potentially) arbitrary discontinuities in the domain. The blurred logarithm, furthermore, is differentiable, and sampling results let us estimate and use its derivatives in our algorithm. This technique is similar to that of randomized smoothing~\cite{Duchi:2012}. {\bf Second}, the negative gradient of the blurred logarithm might point away from the star center (global optimum)---despite all gradients of the unblurred function (when they exist) pointing towards the star center---because of the way blurring interacts with sharp wedge-shaped valleys. Addressing this requires an averaging technique that repeatedly samples Gaussians-in-Gaussians, until it detects sufficient conditions to conclude that the gradient may be used as a cut direction. {\bf Third} and finally, the usual cutting plane criterion---that the volume of the feasible region decreases exponentially---is no longer sufficient in the star-convex setting, as an ellipsoid in two coordinates $(x,y)$ might get repeatedly cut in the $y$ direction without any progress restricting the range of $x$. This is provably not a concern in convex optimization. Our algorithm tackles this issue by ``locking" axes of the ellipsoid smaller than some threshold $\tau$, and seeking cuts orthogonal to these axes; this orthogonal signal may be hidden by much larger derivatives in other directions, and hence requires new techniques to expose. The counterintuitive approach is that, in order to expose structure within an exponentially thin ellipsoid dimension we must search exponentially far outside the ellipsoid in this dimension.

\subsection{Stochastic Optimization}
Our approach extends easily to the stochastic optimization setting.
Here, the objective function is the expectation over a distribution of star-convex functions which share the same star center and optimum value.
For example, in the context of empirical risk minimization with a star-convex loss function $L$ where we assume there is a true hypothesis $h$ such that each labeled example is correctly predicted by $h$, then for each example $i$, the loss $L(\hat{h}(x_i), y_i)$---considered as a function of $\hat{h}$---is star-convex, with global optimum 0 at $\hat{h}=h$. And the overall objective function is the expectation of this quantity, over random examples: $\mathbb{E}_i[L(\hat{h}(x_i), y_i)]$, which is thus also star-convex with global optimum 0 at $\hat{h}=h$.

The optimization algorithm for this new stochastic setting is actually identical to the general star-convex optimization algorithm. Since star-convex functions can take arbitrarily unrelated values on nearby rays from the star center, our algorithm is already fully equipped to deal with the situation, and adding explicitly unrelated values to the model by stochastically sampling from a family of functions $L_i$ makes the problem no harder. The algorithm, analysis, and convergence are unchanged.

\subsection{Outline}
In Section~\ref{sec:Statement} we introduce the oracle model by which we model access to star-convex functions. Because we address such a large class of potentially badly-behaved functions (including, for example, star-convex functions which, in polar coordinates, behave very differently depending on whether the angle is rational or irrational), we take great care to specify the input and output requirements of our optimization algorithm appropriately, in the spirit of Lov\'{a}sz~\cite{Lovasz:1987}.

In Section~\ref{sec:ellipsoid} we present the overall star-convex optimization algorithm, a cutting-plane method that at each stage tracks an ellipsoid that contains the star center (global optimum). The heart of the algorithm consists of finding appropriate cutting planes at each step, which we present in Section~\ref{sec:cuts}.

\section{Problem Statement and Overview}
\label{sec:Statement}

Our aim here is to discuss algorithms for optimizing star-convex functions in the most general setting possible, and thus we must be careful about how the functions are specified to the algorithms.
In particular, we do not assume continuity, so functions with one behavior on the rational points and a separate behavior on irrational points are possible. Therefore, it is essential that our algorithms have access to the function values of $f$ at inputs beyond the usual rational points expressible via standard computer number representations.
As motivation, see Example~\ref{ex:rationals}
which describes a star-convex function that has value 1 on essentially all rational points an algorithm is likely to query, despite having a rather different landscape on the irrational points, leading to a global minimum value of 0.

\begin{example}\label{ex:rationals}
We present the following example of a class of star-convex functions $f$ in 2-dimensions parameterized by integers $(i,j)$ such that $f_{i,j}$ has global optimum at $(1/\sqrt{2}+i,1/\sqrt{3}+j)$.
The class has the property that, if $i$ and $j$ are chosen randomly from an exponential range, then, except with exponentially small probability, any (probabilistic) polynomial time algorithm that accesses $f$ only at rational points will learn nothing about the location of the global optimum.
It demonstrates the need for access beyond the rationals.

We define a star-convex function $f_{i,j}$ parameterized by integers $(i,j)$ that has a unique global optimum at $(1/\sqrt{2}+i,1/\sqrt{3}+j)$.
We evaluate $f(x,y)$ via three cases:
\begin{enumerate}
  \item If the ray from $(1/\sqrt{2}+i,1/\sqrt{3}+j)$ to $(x,y)$ passes through no rational points, then let $f(x,y)=||(x,y)-(1/\sqrt{2}+i,1/\sqrt{3}+j)||_2$.
  \item Otherwise, if the ray from $(1/\sqrt{2}+i,1/\sqrt{3}+j)$ to $(x,y)$ passes through a rational point $(p,q)$, then:
    \begin{enumerate}
    \item If $\lfloor p\rfloor=i$ and $\lfloor q\rfloor=j$ then let $f(x,y)=||(x,y)-(1/\sqrt{2}+i,1/\sqrt{3}+j)||_2$ as above.
    \item Else, let $$f(x,y)=\frac{||(x,y)-(1/\sqrt{2}+i,1/\sqrt{3}+j)||_2}{||(p,q)-(1/\sqrt{2}+i,1/\sqrt{3}+j)||_2}$$
  \end{enumerate}
\end{enumerate}
We note that since $1/\sqrt{2}$ and $1/\sqrt{3}$ are linearly independent over the rationals, no ray from $(1/\sqrt{2}+i,1/\sqrt{3}+j)$ passes through more than one rational point.
Hence the above three cases give a complete definition of $f_{i,j}$.

The function $f$ is star-convex since each of the three cases above apply on the entirety of any ray from the global optimum, and these cases each define a linear function on the ray.
The derivative of $f$ along any of these rays is at most $1/(1-1/\sqrt{2})$, since $(1-1/\sqrt{2})$ is the closest that a point outside the integer square at $(i,j)$ can get to the global optimum $(1/\sqrt{2}+i,1/\sqrt{3}+j)$; thus the function is linearly bounded in any finite region.
Further, $f$ evaluated at any \emph{rational} point $(p,q)$ will fall into Case 2, and thus, unless $(\lfloor p\rfloor,\lfloor q\rfloor)=(i,j)$, the function $f$ will return 1.

Thus, no algorithm that queries $f$ only at rational points can efficiently optimize this class of functions, because unless the algorithm \emph{exactly} guesses the pair $(i,j)$---which was drawn from an exponentially large range---no information is gained (the value of $f$ will always be 1 otherwise).
\end{example}

Since directly querying function values of general star-convex functions at rational points is so limiting, we instead introduce the notion of a \emph{weak sampling evaluation oracle} for a star-convex function, adapting the definition of a \emph{weak evaluation oracle} by Lov\'{a}sz~\cite{Lovasz:1987}.

\begin{definition}
\label{def:evaluation}
A \emph{weak sampling evaluation oracle} for a function $f:\Real^n \to \Real$ takes as inputs a point $x \in \Rational^n$, a positive definite covariance matrix $\Sigma$, and an error parameter $\ferr$. The oracle first chooses a random point $y \from \Normal(x,\Sigma)$, and returns a value $r$ such that $|f(y)-r|<\ferr$.
\end{definition}

The Gaussian sampling in Definition~\ref{def:evaluation} can be equivalently changed to choosing a random point in a ball of desired (small) radius, since any Gaussian distribution can be approximated to arbitrary precision (in the total variation sense) as the convolution of itself with a small enough ball, and vice versa. 

Because inputs and outputs to the oracle must be expressible in a polynomial number of digits, we consider ``well-guaranteed" star-convex functions (in analogy with Lov\'{a}sz~\cite{Lovasz:1987}), where the radius bound $R$ and function bound $B$ below should be interpreted as huge numbers (with polynomial numbers of digits), such as $10^{100}$ and $10^{1000}$ respectively. Numerical accuracy issues in our analysis are analogous to those for the standard ellipsoid algorithm and we do not discuss them further.

\begin{definition}\label{def:boundedness}
  A weak sampling evaluation oracle for a function $f$ is called \emph{well-guaranteed} if the oracle comes with two bounds, $R$ and $B$, such that 1) the global minimum of $f$ is within distance $R$ of the origin, and 2) within distance $10nR$ of the origin, $|f(x)|\le B$.
\end{definition}

Such sampling gets around the obstacles of Examples~\ref{ex:deterministic} and~\ref{ex:rationals} because even if the evaluation oracle is queried at a predictable rational point, the value it returns will represent the function evaluated at an unpredictable, (typically) irrational point nearby.

At this point, our notion of oracle access may \emph{seem} unnatural, in part because it allows access to the function at irrational points that are not computationally expressible.
However, there are two natural and widely used justifications for this approach, of different flavors.
First, the oracle represents a computational model of a mathematical abstraction (the underlying star-convex function), where even for mathematically pathological functions, we can actually in many cases implement simple code to emulate oracle access in the manner described above.
For example, it is in fact easy to implement a weak sampling evaluation oracle for the function of Example~\ref{ex:rationals}, where rays from the star center through rational points behave differently from the other rays.
Since the rational points have Lebesgue measure 0, the complexities introduced by case 2 of Example~\ref{ex:rationals} happen with probability 0, and thus we can write code to simulate a weak sampling evaluation oracle by only implementing the simpler case 1, evaluating $f(x,y)=||(x,y)-(1/\sqrt{2}+i,1/\sqrt{3}+j)||_2$ to the requested accuracy $\pm \ferr$.


Second, setting oracle implementation issues aside, the model cleanly separates \emph{accessing} a potentially pathological function, from \emph{computing properties} of it, in this case, optimizing it. The more pathological a function is, the more surprising it is that any efficient automated technique can extract structure from it. Thus, in some sense, the unrealistic regime of pathological functions, for which we do not even know how to write code emulating oracle access, yields the most surprising regime for the success of the algorithmic results of this paper. This regime is the most insightful from a theory perspective almost exactly because it is the most unnatural and counterintuitive from a practical perspective.



Having established weak sampling evaluation oracles as input to our algorithms, we must now take the same care to define the form of their outputs.
The output of a traditional (convex) optimization problem is an $x$ value such that $f(x)$ is close to the global minimum. For star-convex functions, accessed via a weak sampling evaluation oracle, one can instead ask for a small spherical region in which $f(x)$ remains close to its global optimum---given an input point $x$ within distance $\tau$ of the star center, star convexity and boundedness of $f$ imply that $f(x)$ must be within value $\tau\cdot\frac{4B}{10nR-R-\tau}$ of the global optimum (see the proof of Lemma~\ref{lem:TinyRadius}).

However, such an output requirement is too much to ask of a star-convex optimization algorithm: see Example~\ref{ex:output} for an example of a star-convex function for which there is a large region in which the distributions of $f(x)$ in all small (Gaussian) balls are indistinguishable from each other except with negligible probability, say $10^{-100}$. That is, no algorithm can hope to distinguish a small ball around the global optimum from a small ball anywhere else in this region, except by using close to $10^{100}$ function evaluations. Instead, since on this entire region, the function is close to the global optimum except with $10^{-100}$ probability, the most natural option is for an optimization algorithm to return any portion of this region, namely a region where the function value is within $\ferr$ of optimal except with probability $10^{-100}$. This is how we define the output of a star-convex optimization problem: in analogy to the input convention, we ask optimization algorithms to return a Gaussian region, specified by a mean and covariance, on which the function is near-optimal on all but a $\perr$ fraction of its points. (See Example~\ref{ex:output} for details.)

\begin{definition}\label{def:problem}
  The \emph{weak star-convex optimization problem} for a Lebesgue measurable star-convex function $f$, parameterized by $\perr,\ferr, F$, is as follows.
Given a well-guaranteed weak sampling evaluation oracle for $f$, return with probability at least $1-F$ a Gaussian $\mathcal{G}$ such that $Pr[f(x) \le f^\ast+\ferr:x\leftarrow \mathcal{G}]\ge 1-\perr$.
\end{definition}

We now formally state our main result.
\begin{theorem}\label{thm:main}
Algorithm~\ref{alg:Ellipsoid} optimizes (in the sense of Definition~\ref{def:problem}) any Lebesgue measurable star-convex function $f$ in time $\poly(n,1/\perr,\log\frac{1}{\ferr},\log\frac{1}{F},\log R,\log B)$.
\end{theorem}

Observe that we require only Lebesgue measurability of our objective function, and we make no further continuity or differentiability assumptions.
The measurability is necessary to ensure that probabilities and expectations regarding the function are well-defined, and is essentially the weakest assumption possible for any probabilistic algorithm.
The minimality in our assumptions contrasts that of the work by Nesterov and Polyak, which assumes Lipschitz continuity in the second derivative~\cite{Nesterov:2006}.

It is mathematically interesting that the \emph{cardinality} of the set of star-convex functions which we optimize, even in the 2-dimensional case, equals the huge quantity $2^{|\mathbb{R}|}$, while the cardinality of the entire set of continuous functions---which is \NP-hard to optimize, and strictly contains most standard optimization settings---is only $|\mathbb{R}|$.

\section{Our Optimization Approach}\label{sec:ellipsoid}

In this section, we describe our general strategy of optimizing star-convex functions by an adaptation of the ellipsoid algorithm. Our algorithm, like the standard ellipsoid algorithm, looks for the global optimum inside an ellipsoid whose volume decreases by a fixed ratio for each in a series of iterations, via algorithmically discovered cuts that remove portions of the ellipsoid that are discovered to lie in the ``uphill" direction from the center of the ellipsoid.
In Section~\ref{sec:cuts} we will explore the properties of star-convex functions that will enable us to produce such cuts. However, it is crucial that, unlike for standard convex optimization, such cuts are not enough. Consider, for example, the possibility that for a star-convex function $f(x,y)$, a cutting plane oracle only ever returns cuts in the $x$-direction, and never reduces the size of the ellipsoid in the $y$ direction, a situation which provably cannot occur in the standard convex setting.

Thus, when constructing a cutting plane, our algorithm defines an exponentially small threshold $\tau$ (defined in Definition~\ref{def:Parameters}), such that whenever a semi-principal axis of the ellipsoid is smaller than $\tau$, we guarantee a cut orthogonal to this axis. In this section, we make use of the following characterization of our cutting plane algorithm, Algorithm~\ref{alg:SingleCut}, introduced and analyzed in Section~\ref{sec:cuts}:

\begin{proposition*}[See Proposition~\ref{prop:Alg1Correctness}]
With negligible probability of failure, given an ellipsoid containing the star center, centered at the origin with all semi-principal axes longer than $\tau$ scaled to 1, Algorithm \ref{alg:SingleCut} either a) returns a Gaussian region $\mathcal{G}$ such that $\Pr[f(x) \le f^\ast+\ferr : x \from \mathcal{G}] \ge 1 - \perr$, or b) returns a direction $\vec{d}_\bot$, restricted to the subspace spanned by those ellipsoid semi-principal axes longer than $\tau$, such that when normalized to a unit vector $\hat{\vec{d}}_\bot$, the cut $\{x:x\cdot \hat{\vec{d}}_\bot\leq \frac{1}{3n}\}$ contains the global minimum.
\end{proposition*}

Note that the above precondition on the ellipsoid input is satisfied by applying the appropriate affine transformation to the current ellipsoid before supplying it to Algorithm~\ref{alg:SingleCut}.
Also note that in case a) above, the returned Gaussian region is a solution to the overall optimization problem.

Given these guarantees about the behavior of the cutting plane algorithm, Algorithm~\ref{alg:SingleCut}, we now introduce our overall optimization algorithm, based on the ellipsoid method.

For the following algorithm, we define the number of iterations $m = 6(n+1)\left[n\log\frac{R}{\tau} - (n-1)\log\frac{1+\frac{1}{3n}}{2}\right]$ according to a standard analysis of the multiplicative volume decrease of the ellipsoid method stated in Lemma~\ref{lem:Volume} in Appendix~\ref{ap:ellipsoid}.

\vspace{3mm}
\hspace*{-\parindent}\begin{minipage}{\linewidth}
\begin{framed}
\begin{myalgorithm}[Ellipsoid method]
\label{alg:Ellipsoid}

\noindent\textbf{Input}: A radius $R$ ball centered at the origin which is guaranteed to contain the global minimum.\\
\textbf{Output}: \emph{Either} a) A Gaussian $\mathcal{G}$ such that $\Pr[f(x) \le f^\ast + \ferr: x \from \mathcal{G}] \ge 1 - \perr$ \emph{or} b) An ellipsoid $\mathcal{E}$ such that all function values in $\mathcal{E}$ are at most $f^\ast + \ferr$.
\begin{enumerate}
   \item Let ellipsoid $\mathcal{E}_1$ be the input radius $R$ ball.
   \item For $i \in [1, m+1]$
         \begin{enumerate}
            \renewcommand{\labelenumii}{\arabic{enumi}\alph{enumii}.}
            \item If all the axes of $\mathcal{E}_i$ are shorter than $\tau$, then \textbf{Return} $\mathcal{E}_i$ and \textbf{Halt}.
            \item Otherwise, execute Algorithm~\ref{alg:SingleCut} with ellipsoid $\mathcal{E}_i$.
            \item If it returns a Gaussian, then \textbf{Return} this Gaussian and \textbf{Halt}.
            \item Otherwise, it returns a cut direction $\vec{d}_\bot$.
            Apply an affine transformation such that $\mathcal{E}_i$ becomes the unit ball centered at the origin and $\vec{d}_\bot$ points in the negative $x_1$ direction.
            Use the construction in Lemma~\ref{lem:NewEllipsoid} (see appendix) to construct a new ellipsoid $\mathcal{E}_{i+1}$ that includes the intersection of ellipsoid $\mathcal{E}_i$ with the cut.
            \item If any semi-principal axis of the ellipsoid $\mathcal{E}_{i+1}$ is larger than $3nR$ then apply Lemma~\ref{lem:resizing}, and if the center of the ellipsoid has distance $>R$ from the origin then apply Lemma~\ref{lem:recentering}, to yield an ellipsoid of smaller volume, containing the entire intersection of $\mathcal{E}_{i+1}$ with the ball of radius $R$, that now has all semi-principal axes smaller than $3nR$, and has center in the ball of radius $R$.
          \end{enumerate}
\end{enumerate}
\end{myalgorithm}
\end{framed}
\end{minipage}
\vspace{3mm}

The analysis of Algorithm~\ref{alg:Ellipsoid} is a straightforward adaptation of standard techniques for analyzing the ellipsoid algorithm, bounding the decrease in volume of the feasible ellipsoid at each step, until either the algorithm returns an explicit Gaussian as the solution (in Step 2c), or terminates because the ellipsoid is contained in an exponentially small ball (in Step 2a). See Appendix~\ref{ap:ellipsoid} for full details.

\begin{lemma}
\label{lem:EllipsoidHalt}
Algorithm~\ref{alg:Ellipsoid} halts within $m+1$ iterations either through Step 2a or Step 2c.
\end{lemma}

The following lemma shows that, if the star center (global optimum) is found to lie within a ball of exponentially small radius $\tau$, then this entire ball has function value sufficiently close to the optimum that any point within the ball can be returned.

\begin{lemma}
\label{lem:TinyRadius}
Given an ellipsoid $\mathcal{E}$ that a) has its center within $R$ of the origin, b) is contained within a ball of radius $\tau$, and c) contains the star center $x^\ast$, then $f(x)\in [f^\ast,f^\ast+\ferr]$ for all $x\in \mathcal{E}$.
\end{lemma}

\begin{proof}
For any $x\in\mathcal{E}$, let $y$ be the intersection of the ray $\overrightarrow{x^*x}$ with the sphere of radius $10nR$ about the origin. The boundedness of $f$ implies that $f(y)\leq B$, and that $f(x^\ast)\geq -B$. Thus star-convexity yields that $f(x)\leq f^\ast+2B\frac{||x-x^\ast||_2}{||x^\ast-y||_2}\leq f^\ast+2B\frac{2\tau}{10nR-R-\tau} < f^\ast+\ferr$ since $\tau \ll \ferr$.
\end{proof}

Since the ellipsoid returned in Step 2a will always satisfy the preconditions of Lemma~\ref{lem:TinyRadius}, the entire ellipsoid has function value within $\ferr$ accuracy of the global minimum.
In practice, we can just return this region.
However, if we do wish to conform to our problem statement (Definition~\ref{def:problem}), we may simply return a Gaussian ball of sufficiently small radius and centered at the ellipsoid center, such that there is only negligible probability of sampling a point more than $\tau$ from its center.

In summary, Lemmas~\ref{lem:EllipsoidHalt} and~\ref{lem:TinyRadius} imply that Algorithm~\ref{alg:Ellipsoid} optimizes a star-convex function in time $\poly(n,1/\perr,\log\frac{1}{\ferr},\log\frac{1}{F},\log R,\log B)$, proving Theorem~\ref{thm:main}.

\section{Computing Cuts for Star-Convex Functions}
\label{sec:cuts}

The general goal of this section is to explain how to compute a single cut, in the context of our adapted ellipsoid algorithm explained in Section~\ref{sec:ellipsoid}.
Namely, given (weak sampling, in the sense of Definition~\ref{def:evaluation}) access to a star-convex function $f$, and a bounding ellipsoid in which we know the global optimum lies, we present an algorithm (Algorithm~\ref{alg:SingleCut}) that will either a) return a cut passing close to the center of the ellipsoid and containing the global optimum, or b) directly return the answer to the overall optimization problem.

There are several obstacles to this kind of algorithm that we must overcome.
Star-convex functions may be very discontinuous, without any gradients or even subgradients defined at most points.
Furthermore, arbitrarily small changes in the input value $x$ may produce exponentially large changes in the output of $f$, for any $x$ that is not already exponentially close to the global optimum.
This means that standard techniques to even approximate the function's shape do not remotely suffice in the context of star-convex functions.
Finally, considering again the example of linearly extending an arbitrary function of the unit sphere surface (see Figure~\ref{Fig:fig1} or item 3 in Appendix~\ref{sec:examples}), unlike regular convex functions, star-convex functions do not become ``locally flat" along thin dimensions of increasingly thin ellipsoids.
Thus, the ellipsoid algorithm in general might repeatedly cut certain dimensions while neglecting others.

The algorithm in this section (Algorithm~\ref{alg:SingleCut}) employs three key strategies that we intuitively explain now.

\begin{enumerate}
  \item We do not work directly with the star-convex function $f$, but instead work with a \emph{blurred} version of its \emph{logarithm} (defined in Definition~\ref{def:BlurredLog}), which has a) well-defined derivatives, that b) we can estimate efficiently.
Given a badly-behaved (measurable) function $f$, and the pdf of a Gaussian, $g$, blurring $f$ by the Gaussian produces the convolution $f\ast g$, which has well-behaved derivatives since derivatives commute with convolution: namely, letting $D$ stand for a derivative or second derivative, we have $D(f\ast g)=f\ast (D g)$.
Since the derivatives of a Gaussian pdf $g$ are each bounded, we can estimate derivatives of blurred versions of discontinuous functions by sampling.
Sampling bounds imply that if we have a random variable bounded to a range of size $r$, then we can estimate its mean to within accuracy $\pm\ferr$ by taking the average of $O((\frac{r}{\ferr})^2)$ samples.
The logarithm function (after appropriate translation so that its inputs are positive), maps a star-convex function $f$ to a much smaller range, which enables accurate sampling (in polynomial time).
Therefore, we can efficiently estimate derivatives of the blurred logarithm of $f$.
  \item While intuitively, the negative gradient of (the blurred logarithm of) $f$ points towards the global minimum, this signal might be overpowered by a confounding term in a different direction (see Lemma~\ref{lem:CutLB}), making the negative gradient point \emph{away} from the global optimum in some cases.
To combat this, our algorithm repeatedly estimates the gradient at points sampled from a distribution around the ellipsoid center, and for each gradient, estimates the confounding terms, returning the corresponding gradient only once the confounding terms are found to be small.
Lemma~\ref{lem:DoubleSampling} shows that the confounding terms are small \emph{in expectation}, so Markov's inequality shows that our strategy will rapidly succeed.
  \item The ellipsoid algorithm in general might repeatedly cut certain dimensions while neglecting others, and our algorithm must actively combat this.
If some axes of the ellipsoid ever become exponentially small, then we ``lock" them, meaning that we demand a cut orthogonal to those dimensions, thus maintaining a global lower bound on the length of any axis of the ellipsoid.
Combined with the standard ellipsoid algorithm guarantee that the volume of the ellipsoid decreases exponentially, this implies that the diameter of the ellipsoid can be made exponentially small in polynomial time, letting us conclude our optimization. Finding cuts under this new ``locking" guarantee, however, requires a new algorithmic technique.

We must take advantage of the fact that the ellipsoid is exponentially small along a certain direction $b$ to somehow gain the new ability to produce a cut orthogonal to $b$.
Convex functions become increasingly well-behaved on increasingly narrow regions, however, star-convex functions crucially do not (see item 3 in Section~\ref{sec:examples}).
Thus, as the thickness of the ellipsoid in direction $b$ gets smaller, the structural properties of the function inside our ellipsoid \emph{do not} give us any new powers.
Strangely, we can take advantage of one new parameter regime that at first appears useless: relative to the thinness of the ellipsoid, we have exponentially much space in direction $b$ \emph{outside} of the ellipsoid.
\end{enumerate}

To implement the intuition of Step 1 above, we define $L_z(x)$, a truncated and translated logarithm of the star-convex function $f$, which maps the potentially exponentially large range of $f$ to the (polynomial-sized) range $[\log \epsilon',\log 2B]$, where $\epsilon'$ is defined below (Definition~\ref{def:Parameters}), and is slightly smaller than our function accuracy bound $\ferr$. In the below definition, $z$ intuitively represents our estimate of the global optimum function value, and will record, essentially, the smallest function evaluation seen so far (see Algorithm~\ref{alg:SingleCut}).

\begin{definition}
\label{def:BlurredLog}
Given an objective function $f$ with bound $|f(x)| \le B$ when $||x||\leq R'$ and an offset value $z\ge f^\ast$, we define the \emph{truncated logarithmic version of $f$} to be
$$ L_z(x) = \begin{cases} \log \epsilon' & f(x) - z \le \epsilon'\\ \log 2B & f(x) - z \ge 2B\\ \log (f(x)-z) & \text{ otherwise}\end{cases} $$
\end{definition}

While mapping to a small range, $L_z(x)$ nonetheless gives us a precise view of small changes in the function as we converge to the optimum. The next result shows that, if we ``blur" $L_z(x)$ by drawing $x$ from a Gaussian distribution, then not only can we efficiently estimate the expected value of the ``blurred logarithm of $f$", we can also estimate the derivatives of this expectation with respect to changing either the mean or the variance of the Gaussian.

For an arbitrary bounded (measurable) function $h$, the derivative of its expected value over a Gaussian of width $\sigma$ with respect to either a) moving the center of the Gaussian or b) changing its width $\sigma$, is bounded by $O(\frac{1}{\sigma})$. Thus we normalize the estimates below in terms of the product of the Gaussian width and the derivative, instead of estimating the derivative alone.

\begin{proposition}\label{prop:Sampling}

  Let $\Normal(\mu,\Sigma)$ be a Gaussian with diagonal covariance matrix $\Sigma$ consisting of elements $\sigma_1^2,\sigma_2^2,\ldots,\sigma_n^2$.
  For an error bound $\kappa>0$ and a probability of error $\delta>0$, we can estimate each of the following functions to within error $\kappa$ with probability at least $1-\delta$ using $\poly(n,\frac{1}{\kappa},\log\frac{1}{\delta},\log 2B/\epsilon')$ samples: 1) the expectation $\Exp[L_z(x):x\leftarrow \mathcal{N}(\mu,\Sigma)]$; 2) the (scaled) derivative $\sigma_1\cdot\Diff{\mu_1} \Exp[L_z(x):x\leftarrow \mathcal{N}(\mu,\Sigma)]$; and 3) the derivative with respect to scaling $\sigma_1\cdot\Diff{\sigma_1} \Exp[L_z(x):x\leftarrow \mathcal{N}(\mu,\Sigma)]$. \emph{A fortiori}, these derivatives exist.
\end{proposition}
\begin{proof}

Chernoff bounds yield the first claim, since $L_z\in [\log\epsilon',\log 2B]$, so $O((\frac{\log 2B/\epsilon'}{\kappa})^2\log \frac{1}{\delta})$ samples suffice.

For the second claim, we note that the expectation can be rewritten as the integral over $\mathbb{R}^n$ of $L_z$ times the probability density function of the normal distribution: $\Exp[L_z(x):x\leftarrow \mathcal{N}(\mu,\Sigma)]=(2\pi)^{-n/2}|\Sigma|^{-1/2}\int_{\mathbb{R}^n} L_z(x)\cdot e^{-(x-\mu)^T\Sigma^{-1}(x-\mu)/2}\,\d x$. Thus the derivative of this expression with respect to the first coordinate $\mu_1$ can be expressed by taking the derivative of the probability density function inside the integral, which ends up scaling it by the vector $(x_1-\mu_1)/\sigma_1^2$. Thus $\sigma_1\cdot\Diff{\mu_1} \Exp[L_z(x):x\leftarrow \mathcal{N}(\mu,\Sigma)]=\Exp[\frac{x_1-\mu_1}{\sigma_1}\cdot L_z(x):x\leftarrow \mathcal{N}(\mu,\Sigma)]$. The multiplier $\frac{x_1-\mu_1}{\sigma_1}$ is effectively bounded, because for real numbers $c$, the probability $\Pr[\left|\frac{x_1-\mu_1}{\sigma_1}\right|>c:x\leftarrow\mathcal{N}(\mu,\Sigma)]$ vanishes faster than $\exp(\frac{c}{\sqrt{n}})$. Thus we can pick $c=\poly(n,\log \frac{1}{\kappa},\log 2B/\epsilon')$ such that replacing the expression in the expectation, $\frac{x_1-\mu_1}{\sigma_1}\cdot L_z(x)$, by this expression clamped between $\pm c$ will change the expectation by less than $\frac{\kappa}{2}$. We can thus estimate the expectation of this clamped quantity to within error $\frac{\kappa}{2}$ with probability at least $1-\delta$ via some sample of size $\poly(n,\frac{1}{\kappa},\log\frac{1}{\delta},\log 2B/\epsilon')$, as desired.

The analysis for the third claim (the derivative with respect to $\sigma_1$) is analogous: the derivative of the Gaussian probability density function with respect to $\sigma_1$ (as in the case of a univariate Gaussian of variance $\sigma_1^2$) scales the probability density function by $|x_1-\mu_1|^2/\sigma_1^3$, and after scaling by $\sigma_1$, this expression measures the square of the number of standard deviations from the mean, and as above, the product $\frac{|x_1-\mu_1|^2}{\sigma_1^2} L_z(x)$ can be clamped to some polynomially-bounded interval $[-c,c]$ without changing the expectation by more than $\frac{\kappa}{2}$. The conclusion is analogous to the previous case.
\end{proof}
As mentioned in the previous section, our guiding aim for the ellipsoid method is to prevent any axis of the ellipsoid from getting too small.
Therefore, in order to treat axes differently depending on their length, we shall identify our basis as the unit vectors along the axes of the current ellipsoid and distinguish between axes that are smaller than $\tau$ versus at least $\tau$, where $\tau$ is an exponentially small threshold for ``thinness", defined below in Definition~\ref{def:Parameters}.

\begin{definition}
\label{def:TopBottom}
  Given an ellipsoid, consider an orthonormal basis parallel to its axes. Each semi-principal axis of the ellipsoid whose length is less than $\tau$, we call a ``thin dimension", and the rest are ``non-thin dimensions".
  Given a vector $\mu$, we decompose it into $\mu = \mu_\bot + \mu_\top$ where $\mu_\bot$ is non-zero only in the non-thin dimensions, and $\mu_\top$ is non-zero only in the thin dimensions.
Similarly, given the identity matrix $I$, we decompose it into $I = I_\bot + I_\top$.
\end{definition}
We apply a scaling to the \emph{non-thin} dimensions so as to scale the non-thin semi-principal axes of the ellipsoid to unit vectors (making the ellipsoid a unit ball in the non-thin dimensions).
We keep the thin dimensions as they are.

We present again Proposition~\ref{prop:Alg1Correctness}, describing the guarantees on our cutting plane algorithm required by Algorithm~\ref{alg:Ellipsoid}, and then state the cutting plane algorithm, Algorithm~\ref{alg:SingleCut}.
Below, we make use of constants defined in Definition~\ref{def:Parameters}  that may be interpreted as follows: $k$ is a polynomial number of mesh points;  $\eta$ is the mesh spacing; $\tau'$ is the minimum size of $\sigma_{\top}$, a Gaussian width in the thin dimensions that is somewhat larger than $\tau$, the size of the ellipsoid in the thin dimensions; $\sigma_\bot'$ is a Gaussian width in the $\bot$ (non-thin) dimensions, of inverse polynomial size; $\sigma_\bot$ is polynomially smaller than $\sigma_\bot'$, and $\safety$ is a polynomial quantity. $S$ is a polynomial number of samples defined in the proof of Proposition~\ref{prop:Alg1Correctness}.
Moreover, the function $g$ in the algorithm below is a sum of a probability and two derivatives of the blurred expectation of $L_z(x)$ (the truncated logarithm in Definition~\ref{def:BlurredLog}), and is defined in the statement of Lemma~\ref{lem:MarkovGadget}.

\begin{proposition}[Correctness of Algorithm~\ref{alg:SingleCut}]
\label{prop:Alg1Correctness}
With negligible probability of failure, Algorithm \ref{alg:SingleCut} either a) returns a Gaussian region $\mathcal{G}$ such that $\Pr[f(x) \le f^\ast+\ferr : x \from \mathcal{G}] \ge 1 - \perr$, or b) returns a direction $\vec{d}_\bot$, restricted to the $\bot$ dimensions, such that when normalized to a unit vector $\hat{\vec{d}}_\bot$, the cut $\{x:x\cdot \hat{\vec{d}}_\bot\leq \frac{1}{3n}\}$ contains the global minimum.
\end{proposition}

\vspace{3mm}
\hspace*{-\parindent}\begin{minipage}{\linewidth}
\begin{framed}
\begin{myalgorithm}[Single cut with locked dimensions]
\label{alg:SingleCut}
Take an orthonormal basis for the ellipsoid, as in Definition~\ref{def:TopBottom}.
We apply an affine transformation so that a) the ellipsoid is centered at the origin, and b) the ellipsoid, when restricted to the $\bot$ dimensions, is the unit ball.\\

\noindent\textbf{Input}: An ellipsoid containing the star center, under an affine transformation as above.\\
\textbf{Output}: \emph{Either} a) A cut direction $\vec{d}_\bot$ \emph{or} b) A Gaussian $\mathcal{G}$.

 \begin{enumerate}
 \item For each $i\in [0, k]$
     \begin{enumerate}
     \item[1a.] Evaluate $f$ at $S$ samples from the Gaussian $\mathcal{G}_i$ of width $\tau'\eta^i$ in the $\top$ dimensions and width $\sigma_\bot'$ in the $\bot$ dimensions (that is, $\mathcal{G}_i=\Normal(\vec{0}, \sigma_\bot'^2 I_\bot + \tau'^2\eta^{2i} I_\top)$).
     \item[1b.] If at least $1 - \frac{31\perr}{32}$ fraction of the evaluations are within $\epsilon'$ of the minimum evaluation (at this iteration $i$), then {\bf Return} $\mathcal{G}_i$ and {\bf Halt}.
     \end{enumerate}
 \item Otherwise, let $z$ be the minimum of all samples in Step 1.
\item Repeatedly sample the following, estimating $g$ to within $\pm \frac{\perr}{32}$ each time\vspace{-1mm} $$g(\mu_\bot', \sigma_\top)\, :\; \mu_\bot'\from\Normal(\vec{0}, (\sigma_\bot'^2 - \sigma_\bot^2)I_\bot)\;\text{ and }\;(\sigma_\top = e^X; X \from \Unif[\log \tau', \log R/\safety])$$
\begin{enumerate}\item[3a.]
    Accept the first pair $(\mu_\bot', \sigma_\top)$ such that $g(\mu_\bot', \sigma_\top) > \frac{7}{32}\perr$.
    \end{enumerate}
 \item {\bf Return} the gradient $\vec{d}_\bot = \Grad_\bot[\Exp[L_z(x) : x \from \Normal(\mu_\bot', \sigma_\bot^2 I_\bot + \sigma_\top^2 I_\top)]]$ \quad(the derivative as $\mu_\bot'$ changes, computed via Proposition~\ref{prop:Sampling}).
\end{enumerate}

\end{myalgorithm}
\end{framed}
\end{minipage}\vspace{3mm}

We note that a special case of the above algorithm is when there are no $\top$ (very thin, thinner than $\tau$) dimensions.
This applies, for example, at the beginning of the optimization process.

The rest of this section develops the mathematical analysis leading to the proof of Proposition~\ref{prop:Alg1Correctness}.

The following lemma (Lemma~\ref{lem:CutLB}) analyzes how $L_z(x)$ (the truncated logarithm of $f$) decreases as $x$ moves towards the global optimum, or equivalently, how $L_z(x)$ increases as $x$ moves away.
A crucial complicating factor is that we always average $L_z(x)$ over $x$ drawn from a Gaussian, and moving the mean of a Gaussian away from the global optimum is not exactly the same thing as moving every point in the Gaussian away from the global optimum.
Lemma~\ref{lem:CutLB} expresses the effect of moving the mean away from the global optimum---restricted to the non-thin $\bot$ dimensions---in terms of a positive probability, minus three confounding derivatives.
If we can show that the left hand side of the expression in the following lemma is positive, this means that a cut in the direction of the gradient of (a Gaussian-blurred) $L_z(x)$ in the $\bot$ dimensions is guaranteed to contain the global minimum.
\begin{lemma}
\label{lem:CutLB}
For a star-convex function $f$ with star center (global minimum) at the origin satisfying $f(0) \leq z$, and any mean $\mu=\mu_\bot+\mu_\top$, and variances $\sigma_\bot^2,\sigma_\top^2$,
\begin{align*}
\left.\Diff{\alpha} \Exp\left[L_z(x) : x \from \Normal(\alpha \mu_\bot + \mu_\top, \sigma_\bot^2 I_\bot + \sigma_\top^2 I_\top) \right]\right|_{\alpha = 1}\hspace{-1cm}&\hspace{1cm} \ge \;\Pr[f(x)-z\in (\epsilon',2B) : x \from \Normal(\mu, \sigma_\bot^2 I_\bot + \sigma_\top^2 I_\top)]\\
&-\left.\Diff{\alpha} \Exp\left[L_z(x) : x \from \Normal(\mu_\bot +\alpha \mu_\top, \sigma_\bot^2 I_\bot + \sigma_\top^2 I_\top) \right]\right|_{\alpha = 1}\\
&-\left.\Diff{\alpha} \Exp\left[L_z(x) : x \from \Normal(\mu_\bot + \mu_\top, \alpha^2\sigma_\bot^2 I_\bot + \sigma_\top^2 I_\top) \right]\right|_{\alpha = 1}\\
&-\left.\Diff{\alpha} \Exp\left[L_z(x) : x \from \Normal(\mu_\bot + \mu_\top, \sigma_\bot^2 I_\bot + \alpha^2\sigma_\top^2 I_\top) \right]\right|_{\alpha = 1}
\end{align*}
\end{lemma}

\begin{proof}
Let $f^\ast\leq z$ be the function value at the global minimum $x=0$. By the star-convexity of $f$ with the origin as the star center, $(f(\alpha x)-f^\ast) \ge \alpha (f(x)-f^\ast)$ for all $\alpha > 1$, which implies $(f(\alpha x)-z) \ge \alpha (f(x)-z)$, and thus $\log (f(\alpha x)-z)\ge \log\alpha+\log (f(x)-z)$. Thus the corresponding inequality holds for $L_z(x)$, provided, $f(x)-z\in (\epsilon',2B)$ and $\alpha$ is close enough to 1, so that $L_z$ behaves like $\log f(x)$:
\begin{align*}
\Exp\left[L_z(x) : x \from \Normal(\alpha \mu, \alpha^2\sigma_\bot^2 I_\bot + \alpha^2\sigma_\top^2 I_\top) \right] \geq\;&\Pr[f(x)-z\in (\epsilon',2B) :x \from \Normal(\mu, \sigma_\bot^2 I_\bot + \sigma_\top^2 I_\top)] \cdot   \log \alpha\\&+\Exp\left[L_z(x) : x \from \Normal(\mu, \sigma_\bot^2 I_\bot + \sigma_\top^2 I_\top) \right]\end{align*}

Consider the left hand side as a function $g(\alpha)$; rearranging to put the rightmost term on the left hand side, the inequality says that \begin{equation}\label{eq:prob}g(\alpha)-g(1)\geq \Pr[f(x)-z\in (\epsilon',2B) :x \from \Normal(\mu, \sigma_\bot^2 I_\bot + \sigma_\top^2 I_\top)] \cdot   \log \alpha\end{equation}

By Proposition~\ref{prop:Sampling}, $g$ has a derivative at $\alpha=1$, which equals $\lim_{\alpha\to 1}\frac{g(\alpha)-g(1)}{\alpha-1}=\lim_{\alpha\to 1}\frac{g(\alpha)-g(1)}{\log \alpha}$ (by L'H\^{o}pital's rule), which by Equation~\ref{eq:prob} is thus greater than or equal to the probability $\Pr[f(x)-z\in (\epsilon',2B) :x \from \Normal(\mu, \sigma_\bot^2 I_\bot + \sigma_\top^2 I_\top)]$.

Therefore
$$ \left.\Diff{\alpha} \Exp\left[L_z(x) : x \from \Normal(\alpha \mu, \alpha^2\sigma_\bot^2 I_\bot + \alpha^2\sigma_\top^2 I_\top) \right]\right|_{\alpha = 1} \ge \Pr[f(x)-z\in (\epsilon',2B)  : x \from \Normal(\mu, \sigma_\bot^2 I_\bot + \sigma_\top^2 I_\top)] $$
Rewriting the left hand side of the above inequality into the sum of the four derivatives in the lemma statement gives the result required.
\end{proof}

As explained above, if we are able to lower bound the left hand side of Lemma~\ref{lem:CutLB}, then we will be able to make a cut to the ellipsoid that contains the global optimum, while being perpendicular to the thin dimensions.
Therefore, we need to upper bound the three derivatives on the right hand side of the inequality of Lemma~\ref{lem:CutLB}.
Lemmas~\ref{lem:MovingInThinDimension}, \ref{lem:Isoperimetric}, and \ref{lem:SigmaTopExpectation} bound these three terms respectively.

In the following, we define the radius $\tau'=\tau\frac{16}{\perr}\cdot\log\left(\frac{2B}{\epsilon'}\right)\frac{2\sqrt{2}}{\sqrt{\pi}}$ to be slightly larger than $\tau$ (see Definition~\ref{def:Parameters} below for details).
Lemma~\ref{lem:MovingInThinDimension} considers widths $\sigma_\top\geq\tau'$, in line with the operation of our Algorithm~\ref{alg:SingleCut}.

The next two lemmas directly bound the first two derivative terms on the right in the expression of Lemma~\ref{lem:CutLB}, via a direct calculation of how fast the average of an arbitrary function can change with respect to the Gaussian parameters.

\begin{lemma}
\label{lem:MovingInThinDimension}
For all $z, \mu_\bot, \sigma_\bot$ and all $||\mu_\top||_2\leq\tau$ and $\sigma_\top\geq\tau'$ we have:
$$ \left|\left.\Diff{\alpha} \Exp\left[L_z(x) : x \from \Normal(\mu_\bot +\alpha \mu_\top, \sigma_\bot^2 I_\bot + \sigma_\top^2 I_\top) \right]\right|_{\alpha = 1}\right| \le \frac{\perr}{16}
$$
\end{lemma}

\begin{proof}
Given $\mu_\bot$, $\sigma_\bot$ and $\sigma_\top \ge \tau'$, define $P_{\mu_\top}$ to be the probability measure of the distribution $\Normal(\mu_\bot + \mu_\top, \sigma_\bot^2 I_\bot + \sigma_\top^2 I_\top)$. We have
\begin{align*}
\text{L.H.S.} &= \lim_{\Delta_\alpha \to 0} \frac{1}{|\Delta_\alpha|} \left[\left|\int_{\Real^n} L_z(x) (\d P_{(1+\Delta_\alpha)\mu_\top} - \d P_{\mu_\top})\right|\right]\\
&\le\lim_{\Delta_\alpha \to 0} \frac{1}{|\Delta_\alpha|} \left[\int_{\Real^n} \left|L_z(x)\right| \left|\d P_{(1+\Delta_\alpha)\mu_\top} - \d P_{\mu_\top}\right|\right]\\
&\le\lim_{\Delta_\alpha \to 0} \frac{1}{|\Delta_\alpha|} \left[\int_{\Real^n} \log\left(\frac{2B}{\epsilon'}\right) \left|\d P_{(1+\Delta_\alpha)\mu_\top} - \d P_{\mu_\top}\right|\right] \text{ since $[\log\epsilon',\log 2B]$ contains both $L_z(x)$ and 0}\\
&\le\lim_{\Delta_\alpha \to 0} \frac{1}{|\Delta_\alpha|} \log\left(\frac{2B}{\epsilon'}\right) \frac{2\sqrt{2}}{\sqrt{\pi}}\frac{||\Delta_\alpha\mu_\top||_2}{\sigma_\top} \text{\quad (total variation distance between Gaussians)}\\
&= \log\left(\frac{2B}{\epsilon'}\right)\frac{2\sqrt{2}}{\sqrt{\pi}}\frac{||\mu_\top||_2}{\sigma_\top} \quad\le \frac{\perr}{16} \text{\quad by the definition of $\tau'$}
\end{align*}
\end{proof}

\begin{lemma}
\label{lem:Isoperimetric}
For any $z, \mu, \sigma_\bot$ and $\sigma_\top$ we have the following inequality:
\begin{align*}
\left|\left.\Diff{\alpha} \Exp\left[L_z(x) : x \from \Normal(\mu, \alpha^2\sigma_\bot^2 I_\bot + \sigma_\top^2 I_\top) \right]\right|_{\alpha = 1}\right| &\le \sqrt{2\dim(\bot)}\log\left(\frac{2B}{\epsilon'}\right)\\
&\le \sqrt{2n}\log\left(\frac{2B}{\epsilon'}\right)
\end{align*}
where $\dim(\bot)$ is the number of $\bot$ dimensions.
\end{lemma}

\begin{proof}
This proof is analogous to that of Lemma \ref{lem:MovingInThinDimension}, making use of Pinsker's inequality to bound the total variation distance between Gaussians of different variances, via their KL divergence.
\end{proof}

For our overall strategy of bounding each of the three terms in the right hand side of the inequality of Lemma~\ref{lem:CutLB}: the bound of Lemma~\ref{lem:MovingInThinDimension} we use as is, however the bound of Lemma~\ref{lem:Isoperimetric} is somewhat larger than 1, and needs to be improved (in order to ultimately compare it with $\perr$). We accomplish this with an averaging argument, where we sample the mean of our Gaussian from a somewhat larger Gaussian, which will effectively decrease the expectation of the left hand side of the expression in Lemma~\ref{lem:Isoperimetric} by the ratio of the variances of the two Gaussians. We analyze the effect of this ``double-sampling" process on the relevant quantities in Lemma~\ref{lem:DoubleSampling}.
\begin{lemma}
\label{lem:DoubleSampling}
For all $z, \mu_\top, \sigma_\top, \mu_\bot$ and $\sigma_\bot < \sigma_\bot'$ we have the following two equalities:
\begin{align*}
&\Exp[\Pr[f(x)-z \in (\epsilon',2B) : x \from \Normal(\mu_\bot' + \mu_\top, \sigma_\bot^2 I_\bot + \sigma_\top^2 I_\top)] : \mu_\bot' \from \Normal(\mu_\bot, (\sigma_\bot'^2 - \sigma_\bot^2)I_\bot)]\\
=\;&\Pr[f(x)-z \in (\epsilon',2B) : x \from \Normal(\mu_\bot + \mu_\top, \sigma_\bot'^2 I_\bot + \sigma_\top^2 I_\top)]
\end{align*}
and\begin{align*}
&\Exp\left[\left.\Diff{\alpha} \Exp\left[L_z(x) : x \from \Normal(\mu_\bot' + \mu_\top, \alpha^2\sigma_\bot^2 I_\bot + \sigma_\top^2 I_\top) \right]\right|_{\alpha = 1} : \mu_\bot' \from \Normal(\mu_\bot, (\sigma_\bot'^2 - \sigma_\bot^2)I_\bot)\right]\\
=\;&\left(\frac{\sigma_\bot}{\sigma_\bot'}\right)^2\left.\Diff{\alpha} \Exp\left[L_z(x) : x \from \Normal(\mu_\bot + \mu_\top, \alpha^2\sigma_\bot'^2 I_\bot + \sigma_\top^2 I_\top) \right]\right|_{\alpha = 1}
\end{align*}
\end{lemma}

\begin{proof}
Because adding two Gaussian random variables produces a Gaussian random variable, the distribution
$$ (x \from \Normal(\mu_\bot' + \mu_\top, \alpha^2\sigma_\bot^2 I_\bot + \sigma_\top^2 I_\top) : \mu_\bot' \from \Normal(\mu_\bot, (\sigma_\bot'^2 - \sigma_\bot^2)I_\bot)) $$
is equivalent to the distribution $\Normal(\mu_\bot + \mu_\top, (\sigma_\bot'^2 + (\alpha^2-1)\sigma_\bot^2)I_\bot + \sigma_\top^2 I_\top)$, which gives the first equality.
To prove the second equality, we move the derivative outside the expectation, and combine the two expectations into the equivalent expression \[
 \Exp\left[L_z(x) : x \from \Normal(\mu_\bot + \mu_\top, (\sigma_\bot'^2+(\alpha^2-1)\sigma_\bot^2) I_\bot + \sigma_\top^2 I_\top) \right]
\]
Further observe that the derivatives of the $\alpha$ dependencies in the two expressions we are comparing are $\left.\Diff{\alpha} (\sigma_\bot'^2 + (\alpha^2-1)\sigma_\bot^2)\right|_{\alpha = 1}$ and $\left.\Diff{\alpha} \alpha^2 \sigma_\bot'^2\right|_{\alpha = 1}$, which have ratio exactly $(\sigma_\bot/\sigma_\bot')^2$, giving the second equality.
\end{proof}

Having bounded the first two ``confounding derivatives" from the right hand side of Lemma~\ref{lem:CutLB}, we now bound the third. We cannot directly bound this derivative, so we again employ an averaging argument. Intuitively, this term records how $L_z$ increases as the width $\sigma_\top$ increases; however, since $L_z$ is bounded between $\ferr'$ and $2B$, the derivative cannot stay large over a large range of $\sigma_\top$. Crucially, we let $\sigma_\top$ vary over a huge range outside the ellipsoid, between $\tau'>\tau$ and $R/\safety$ (with $\safety$ a polynomial factor, defined in Definition~\ref{def:Parameters}).
In order to give $\sigma_\top$ a large enough range for the following bound to be meaningful, $\tau'$ (and hence $\tau$, the threshold for ``thinness") has to be exponentially small, as specified below in Definition~\ref{def:Parameters}.
\begin{lemma}
\label{lem:SigmaTopExpectation}
For all $z, \mu, \sigma_\bot$,
\begin{align*}
&\Exp\left[\left.\Diff{\alpha} \Exp\left[L_z(x) : x \from \Normal(\mu_\bot + \mu_\top, \sigma_\bot^2 I_\bot + \alpha^2\sigma_\top^2 I_\top) \right]\right|_{\alpha = 1} : (\sigma_\top = e^X; X \from \Unif[\log \tau', \log R/\safety])\right]\\
\le\;&\frac{\log 2B - \log \epsilon'}{\log (R/\safety) - \log \tau'}
\end{align*}
\end{lemma}

\begin{proof}
The probability density function of $\sigma_\top$ is $f_{\sigma_\top}(y) = \frac{1}{y}\cdot\frac{1}{\log (R/\safety) - \log \tau'}$.
Let $$ h(\sigma_\top) = \Exp\left[L_z(x) : x \from \Normal(\mu_\bot + \mu_\top, \sigma_\bot^2 I_\bot + \sigma_\top^2 I_\top) \right]. $$
The following quantity thus equals the left hand side of the inequality in the lemma statement without the scaling factor $\frac{1}{\log (R/\safety) - \log \tau'}$:
$$ \int_{\tau'}^{R/\safety} \frac{1}{y} \left.\Diff{\alpha} h(\alpha y)\right|_{\alpha = 1} \; \d y = \int_{\tau'}^{R/\safety} \Diff{y} h(y) \; \d y = h(R/\safety) - h(\tau') \le \log 2B - \log \epsilon' $$
since $L_z(x)$ and hence $h(\sigma_\top)$ is bounded between $\log \epsilon'$ and $\log 2B$.
The lemma statement follows.
\end{proof}

Recalling that in Lemma~\ref{lem:MovingInThinDimension} we satisfactorily bounded the second term from the right hand side of the inequality of Lemma~\ref{lem:CutLB}, we now combine the results of the previous lemmas so that we may use Markov's inequality to bound the sum of the remaining three terms. The following lemma makes an assumption that the probability that $f(x)-z \notin (\ferr',2B)$ is bounded away from 1; in the context of our Algorithm~\ref{alg:SingleCut}, the case where this assumption is false turns out to imply that our algorithm has already successfully optimized the function---despite this algorithm being intended merely to seek another cut. This result is shown in Lemma~\ref{lem:Victory}.

\begin{lemma}
\label{lem:MarkovGadget}
Given $z, \mu_\top$ and $\sigma_\bot$, define an auxiliary function
\begin{align*}
g(\mu_\bot', \sigma_\top) &= \Pr[f(x)-z \in (\epsilon',2B) : x \from \Normal(\mu_\bot' + \mu_\top, \sigma_\bot^2 I_\bot + \sigma_\top^2 I_\top)]\\
&-\left.\Diff{\alpha} \Exp\left[L_z(x) : x \from \Normal(\mu_\bot' + \mu_\top, \alpha^2\sigma_\bot^2 I_\bot + \sigma_\top^2 I_\top) \right]\right|_{\alpha = 1}\\
&- \left.\Diff{\alpha} \Exp\left[L_z(x) : x \from \Normal(\mu_\bot' + \mu_\top, \sigma_\bot^2 I_\bot + \alpha^2\sigma_\top^2 I_\top) \right]\right|_{\alpha = 1}
\end{align*}
Suppose we have $\mu_\bot$ and $\sigma_\bot' > \sigma_\bot$ such that for all $\sigma_\top\in [\tau',R/\safety]$ we have $\Pr[f(x)-z \in (\epsilon',2B) : x \from \Normal(\mu_\bot + \mu_\top, \sigma_\bot'^2 I_\bot + \sigma_\top^2 I_\top)] \ge \frac{22}{32}\perr$.
Thus, letting $D$ be the joint distribution of independent random variables $\mu_\bot'\from\Normal(\mu_\bot, (\sigma_\bot'^2 - \sigma_\bot^2)I_\bot)$ and $(\sigma_\top = e^X; X \from \Unif[\log \tau', \log R/\safety])$, we have that
$$ \Pr[g(\mu_\bot', \sigma_\top) > \frac{1}{2}E : (\mu_\bot', \sigma_\top) \from D] \ge \frac{E}{2(1+2\sqrt{2n}\log\frac{2B}{\epsilon'})} $$
where
$$ E = \frac{22}{32}\perr - \left(\frac{\sigma_\bot}{\sigma_\bot'}\right)^2\sqrt{2n}\log\left(\frac{2B}{\epsilon'}\right) - \frac{\log 2B - \log \epsilon'}{\log (R/\safety) - \log \tau'}$$
\end{lemma}

\begin{proof}
We bound the expectations of the individual terms in $g(\mu_\bot,\sigma_\top)$ when we draw $(\mu_\bot,\sigma_\top)\from D$.
For the first term, we make use of Lemma~\ref{lem:DoubleSampling} to simplify the double expectations.
\begin{align*}
&\Exp[\Pr[f(x)-z \in (\epsilon',2B) : x \from \Normal(\mu, \sigma_\bot^2 I_\bot+\sigma_\top^2 I_\top)] : (\mu_\bot, \sigma_\top) \from D]\\
=\;&\Exp[\Pr[f(x)-z \in (\epsilon',2B) : x \from \Normal(\mu_\bot' + \mu_\top, \sigma_\bot'^2 I_\bot + \sigma_\top^2 I_\top)] : (\sigma_\top = e^X; X \from \Unif[\log \tau', \log (R/\safety)])]\\
\ge\;& \frac{22}{32}\perr \text{\quad by assumption.}
\end{align*}
The other two terms are bounded by Lemmas \ref{lem:Isoperimetric} (via Lemma~\ref{lem:DoubleSampling}) and \ref{lem:SigmaTopExpectation} respectively, yielding that $\Exp[g(\mu_\bot, \sigma_\top):(\mu_\bot, \sigma_\bot) \from D]\ge E$. Having bounded the expectation, we now upper bound $g$ so that we can apply Markov's inequality. The first term of $g$ is a probability and hence is upper bounded by 1; the next two terms are each upper bounded via Lemma~\ref{lem:Isoperimetric} by $\sqrt{2n}\log\frac{2B}{\epsilon'}$. Thus $g(\mu_\bot, \sigma_\top)\leq 1+2\sqrt{2n}\log\frac{2B}{\epsilon'}$.

Markov's inequality yields that $\Pr[g(\mu_\bot, \sigma_\top)>\frac{E}{2}]>\frac{E}{2(1+2\sqrt{2n}\log\frac{2B}{\epsilon'})}$.
\end{proof}

Using Lemma~\ref{lem:MarkovGadget} and the previous lemmas, we have successfully upper bounded the confounding derivatives and hence lower bounded the left hand side of Lemma~\ref{lem:CutLB}, but only when the probability term is bounded away from 0 for all $\sigma_\top$.
As we mentioned above, in fact, if the probability term is too small for some $\sigma_\top$, then it turns out we already have a Gaussian region that can be returned as the optimization output.
Intuitively, if $f$ is ``flat'' (except on a small fraction of the points) on a large Gaussian region that is known to be near the global minimum (``near", relative to the size of the Gaussian), then the function value on this Gaussian is essentially the global minimum, and thus this Gaussian region may be returned as the overall answer to the optimization problem.
The next two lemmas formalize this result.

\begin{lemma}\label{lem:tail}
  Given a probability distribution on the positive real line with pdf proportional to $f(x)=e^{-(x-c)^2/2}x^{n-1}$, for $c\geq 0$ then any set with probability at least $0.9$ under this distribution contains two points with ratio at least $1+1/(6\max\{c,\sqrt{n}\})$.
\end{lemma}
\begin{proof}

  The pdf has maximum at $x$ value $m=\frac{c+\sqrt{c^2+4(n-1)}}{2}$ (as can be found by differentiating the logarithm of $f$). We upper bound $f$ by $f(x)\leq f(m)\cdot e^{-(x-m)^2/2}$, since, taking the logarithm of both sides, the function value and first derivative match at $x=m$ (the derivative is 0), while the second derivative of the (logarithm of the) right hand side equals $-1$ while the left hand side yields \begin{equation}\label{eq:logf}\frac{\d^2}{\d x^2}\log f(x) =-1-\frac{\d^2}{\d x^2}\log x^{n-1}=-1-\frac{n-1}{x^2}\end{equation} which is smaller than $-1$ for all $x$. Thus the integral of $f$ from 0 to $\infty$ is at most $f(m)$ times the integral of the Gaussian $e^{-(x-m)^2/2}$, which is at most $f(m)\sqrt{2\pi}$.

  We correspondingly lower bound $f(x)$. We note the value $m$ (specifying the $x$ coordinate of the maximum of $f$) is bounded as $m=\frac{c+\sqrt{c^2+4(n-1)}}{2}\geq \sqrt{n-1}$. Thus for $x\geq m$, Equation~\ref{eq:logf} yields the bound $\frac{\d^2}{\d x^2}\log f(x)\geq -2$. Analogously to above this yields the lower bound that $f(x)\geq f(m) e^{-(x-m)^2}$ for $x\geq m$---where the standard deviation of this Gaussian is now $1/\sqrt{2}$ instead of 1.

  Given this lower bound $f(x)\geq f(m) e^{-(x-m)^2}$ for $x\geq m$ and the upper bound $f(m)\sqrt{2\pi}$ on the overall integral of $f$, the probability of $x$ lying in the interval $[m,m+\frac{1}{3}]$ is at least $$\frac{1}{f(m)\sqrt{2\pi}}\int_{m}^{m+1/3} f(m) e^{-(x-m)^2}\,\d x > 0.1$$ Similarly, the probability of $x$ lying in the interval $[m+\frac{2}{3},\infty)$ is greater than $0.1$.

  Thus, any set $S$ of probability at least $0.9$ on the distribution with pdf proportional to $f$ must contain points from \emph{both} intervals $[m,m+\frac{1}{3}]$ and $[m+\frac{2}{3},\infty)$. Hence $S$ must contain two points with ratio at least $\frac{m+2/3}{m+1/3}=1+\frac{1}{3m+1}$.

  We conclude the lemma by bounding $3m+1$ in terms of $\max\{c,\sqrt{n}\}$. Let $y=\max\{c,\sqrt{n}\}$, which is at least $1$ for $n\geq 1$. We have $3m+1=1+\frac{3}{2}(c+\sqrt{c^2+4(n-1)})\leq 1+\frac{3}{2}(y+\sqrt{5y^2-4})$, which is easily seen to be at most $6y$ for $y\geq 1$. Thus $S$ contains two points with ratio at least $1+\frac{1}{3m+1}\geq 1+1/(6\max\{c,\sqrt{n}\})$, as desired.

\end{proof}

\begin{lemma}
\label{lem:Victory}
Given a star-convex function $f$ with global optimum at the origin, if for some location $\mu$ we have $\Pr[f(x) \in [z, z+\ferr'] : x \from \Normal(\mu, I)] > 0.95$ then
the function value at the global optimum, $f^\ast=f(0)$, satisfies $f^\ast\geq z-6\epsilon'\max\{||\mu||_2,\sqrt{n}\}$.

\end{lemma}
\begin{proof}
  Consider a ray through the global minimum (at the origin), and let the ray be defined as all positive multiples of a unit vector $x$. Consider the set $C_x\subset\mathbb{R}^+$ of scaling factors $r$ such that $f(rx)\in[z,z+\ferr']$. We note that if $C_x$ contains two values $r<r'$ with some ratio $\kappa=r'/r$ then by the star-convexity of $f$ (on this ray), the global minimum of $f$ must have value at least $z-\ferr'/(\kappa-1)$. Thus we want to show that there exists a direction $x$ with a set $C_x$ that contains two widely--spaced elements; namely, it is impossible for each $C_x$ to only contain values within small ratio of each other.

  We note that we may express $Pr[f(x)\in[z,z+\ferr']:x\leftarrow\mathcal{N}(\mu,I)]$ in terms of the sets $C_x$ as an integral in polar coordinates. Letting $\mathbb{S}^{n-1}$ denote the $(n-1)$-dimensional sphere, this probability equals the ratio
  \begin{equation}\label{eq:integralRatio} \left.\int_{x\in \mathbb{S}^{n-1}} \int_{r\in C_x} r^{n-1}e^{-\frac{|rx-\mu|^2}{2}}\,\d r\,\d x\right/\int_{x\in \mathbb{S}^{n-1}} \int_{r\in \mathbb{R}^+} r^{n-1}e^{-\frac{|rx-\mu|^2}{2}}\,\d r\,\d x \end{equation}
  Consider those directions $x$ with a positive component in the direction of $\mu$; this halfspace $H$ comprises at least half the probability mass of the Gaussian $\mathcal{N}(\mu,I)$, and thus, within this halfspace, $\Pr[f(x) \in [z, z+\ferr'] : x \from \Normal(\mu, I) \,|\, x\in H] > 0.9$. Suppose for the sake of contradiction that for each $x\in H$, defining $c<||\mu||_2$  to be the component of $\mu$ in the direction of $x$ (where $c\geq 0$), the set $C_x$ does not contain any points of ratio at least $1+1/(6\max\{c,\sqrt{n}\})$. Thus by Lemma~\ref{lem:tail}, we have the bound $\left.\int_{r\in C_x} r^{n-1}e^{-(r-c)^2/2}\,\d r\right/\int_{r\in \mathbb{R}^+} r^{n-1}e^{-(r-c)^2/2}\,\d r< 0.9$. For each $x$, with $c$ as defined, the ratio between these integrals equals the ratio between the corresponding inner integrals of Equation~\ref{eq:integralRatio}, $ \left.\int_{r\in C_x} r^{n-1}e^{-\frac{|rx-\mu|^2}{2}}\,\d r\right/\int_{r\in \mathbb{R}^+} r^{n-1}e^{-\frac{|rx-\mu|^2}{2}}\,\d r$, yielding that the average value of this ratio over $x$ in the halfspace $H$ is at most $0.9$. This contradicts the fact derived earlier that $\Pr[f(x) \in [z, z+\ferr'] : x \from \Normal(\mu, I) \,|\, x\in H] > 0.9$.

  Therefore there must exist a direction $x$ such that $C_x$ contains two points of ratio at least $1+1/(6\max\{c,\sqrt{n}\})$, where since $c\leq ||\mu||_2$, this ratio is thus at least $1+1/(6\max\{||\mu||_2,\sqrt{n}\})$. We then conclude $f^\ast\geq z-6\epsilon'\max\{||\mu||_2,\sqrt{n}\}$.
\end{proof}

We now wish to combine Lemmas~\ref{lem:MarkovGadget} and~\ref{lem:Victory}, in the sense that we reason either one or the other applies.
The na\"{i}ve interpretation of the above would involve algorithmically checking the condition of Lemma~\ref{lem:Victory} for each and every Gaussian induced by all the $\sigma_\top$ in the continuous (exponential) range as in Lemma~\ref{lem:MarkovGadget}, which is obviously impossible.
Therefore, we need to choose a mesh in that range such that the total variation distance between a Gaussian with a width $\sigma_\top$ in the $\top$ dimensions and the closest Gaussian in the mesh is upper bounded by some quantity proportional to $\perr$.
This allows us to reason that after sampling all points in the mesh, either Lemma~\ref{lem:Victory} applies, in which case we can conclude our optimization, or we can apply Lemma~\ref{lem:MarkovGadget} to produce a cut.
The following lemma shows that a geometric spacing in the mesh suffices, thus explaining Step 1 of Algorithm~\ref{alg:SingleCut}.

\begin{lemma}
\label{lem:Pinsker}
Given $\mu$ and $\sigma_\bot$, let $P_{\sigma_\top}$ be the probability measure of the distribution $\Normal(\mu, \sigma_\bot^2 I_\bot + \sigma_\top^2 I_\top)$.
For $\sigma_\top < \sigma_\top'$, we have the inequalities
$$ \frac{1}{2}|P_{\sigma_\top} - P_{\sigma_\top'}| \le \sqrt{\frac{1}{2}D_{KL}(P_{\sigma_\top} || P_{\sigma_\top'})} \le \sqrt{\frac{\dim(\top)}{2}\log\left(\frac{\sigma_\top'}{\sigma_\top}\right)} \le \sqrt{\frac{n}{2}\log\left(\frac{\sigma_\top'}{\sigma_\top}\right)} $$
where $\frac{1}{2}|P_{\sigma_\top} - P_{\sigma_\top'}|$ denotes the \emph{total variation distance} and $D_{KL}$ denotes the \emph{KL divergence}.
That is, if $\frac{\sigma_\top'}{\sigma_\top} \le e^{\frac{2p^2}{n}}$, then the difference in probabilities of event $E$ happening under the two distributions is at most $p$.
\end{lemma}

\begin{proof}
The first inequality is Pinsker's inequality.
A direct calculation from the standard expression for the KL-divergence of multivariate Gaussians shows the second inequality.
The third inequality is self-evident.
\end{proof}

We now have all the theoretical tools to reason about Algorithm~\ref{alg:SingleCut}.
In the following definition, we choose all the quantities we have used in the paper.
\begin{definition}\label{def:Parameters}
\begin{alignat*}{4}
\text{Let }&\safety = \sqrt{n}\left(1+ \sqrt{\frac{4}{3}}\sqrt{n + \frac{1}{\perr} + \log \frac{1}{\ferr} + \log B + \log R + \log \frac{1}{F}}\right) \hspace{1cm}&& \hspace{2cm} \sigma_\bot' = \frac{1}{3n\safety} \hspace{1cm} \hspace{-2cm}&& \\
&\ferr' = \ferr\left(1 + \frac{12}{\sigma_\bot'}\right)^{-1} &&\hspace{-5cm}\sigma_\bot = \sigma_\bot'\sqrt{\frac{\perr/8}{\log 2B - \log \ferr'}\sqrt{\frac{1}{2n}}} && \hspace{-1cm}\tau' = \frac{R}{\safety}\left(\frac{2B}{\ferr'}\right)^{-\frac{16}{\perr}}\\
&\tau=\tau'\frac{\perr}{16}\cdot\frac{1}{\log\left(\frac{2B}{\epsilon'}\right)}\frac{\sqrt{\pi}}{2\sqrt{2}} &&\hspace{-5cm} \eta = e^{\frac{\perr^2}{8n}} && \hspace{-1cm} k = \frac{\log \left(\frac{2B}{\ferr'}\right)^{\frac{16}{\perr}}}{\log \eta}
\end{alignat*}
\end{definition}

Finally, we present the proof of Proposition~\ref{prop:Alg1Correctness} (restated below for convenience), which establishes the correctness of Algorithm~\ref{alg:SingleCut}.

\medskip\noindent{\bf Proposition~\ref{prop:Alg1Correctness}} (Correctness of Algorithm~\ref{alg:SingleCut}){\bf .} \emph{With negligible probability of failure, Algorithm \ref{alg:SingleCut} either a) returns a Gaussian region $\mathcal{G}$ such that $\Pr[f(x) \le f^\ast+\ferr : x \from \mathcal{G}] \ge 1 - \perr$, or b) returns a direction $\vec{d}_\bot$, restricted to the $\bot$ dimensions, such that when normalized to a unit vector $\hat{\vec{d}}_\bot$, the cut $\{x:x\cdot \hat{\vec{d}}_\bot\leq \frac{1}{3n}\}$ contains the global minimum.
}

\begin{proof}
To keep the proof simpler, we now translate the coordinate system to place the \emph{global optimum} at the origin. If the probability error parameter $\perr\ge\frac{1}{20}$, then make $\perr=\frac{1}{21}$, which can only improve the results.

If the algorithm halts in Step 1b, this means that for the returned Gaussian $\mathcal{G}_i$,  at least $1-\frac{31}{32}\perr$ fraction of the samples were within a range $[z_i,z_i+\ferr']$, denoting by $z_i$ the smallest observed sample from $\mathcal{G}_i$. Chernoff bounds imply that (except with probability of failure $F$), the true probability $\Pr[f(x)\in[z_i,z_i+\ferr']:x\from\mathcal{G}_i]$ is at least $1-\perr$, provided we take $S=\poly(1/\delta,|\log F|)$ samples.

Thus we apply Lemma~\ref{lem:Victory}, since $\delta<\frac{1}{20}$, and conclude that $f^\ast\geq z_i-6\epsilon'\max\{||\mu||_2,\sqrt{n}\}$, where we must now bound ``$\mu$", which from the notation of Lemma~\ref{lem:Victory} means the ``distance between the global minimum and the center of the ellipsoid, measured by the number of standard deviations of Gaussian $\mathcal{G}_i$". Specifically, in the $\bot$ dimensions, the ellipsoid has radius 1 and the Gaussian $\mathcal{G}_i=\Normal(0, \sigma_\bot'^2 I_\bot + \tau'^2\eta^{2i} I_\top)$ has radius $\sigma_\bot'$; in the $\top$ dimensions, the Gaussian has radius larger than the ellipsoid, so the distance in these directions is less than 1, which is certainly less than $\frac{1}{\sigma_\bot'}$. Thus in total we bound $||\mu||_2\leq \frac{2}{\sigma_\bot'}$. Plugging this bound into the conclusion of Lemma~\ref{lem:Victory}, where $||\mu||_2\geq\sqrt{n}$, we have $f^\ast\geq z_i-6\epsilon'\frac{2}{\sigma_\bot'}\geq z_i+\ferr'-\ferr$, by the definition of $\ferr'$. Thus, the returned Gaussian satisfies the desired properties of the output of our optimization algorithm: $\Pr[f(x) \in [f^\ast, f^\ast+\ferr] : x \from \mathcal{G}_i] \ge 1-\delta$.

We now analyze the situation when the algorithm does \emph{not} halt in Step 1b. In this case, for each $i$, the proportion of samples larger than $z_i+\ferr'$ was observed to be at least $\frac{31}{32}\perr$. Since $z=\min_i z_i$, the proportion larger than $z+\ferr'$ must also be at least $\frac{31}{32}\perr$. As above, by Chernoff bounds, for each $i$ (except with probability of failure $F$), the true probability $\Pr[f(x)>z+\ferr':x\from\mathcal{G}_i]$ is at least $\frac{30}{32}\perr$, provided we take $S=\poly(1/\delta,|\log F|)$ samples.

Since as $i$ ranges from $0$ to $k$, the Gaussians $\mathcal{G}_i$ vary exponentially in their width in the $\top$ dimensions, as $\sigma_{\top,i}=\tau'\eta^i$, these widths form a fine (exponentially spaced) mesh over this entire region from $\tau'$ to $\frac{R}{\safety}$. Thus, for \emph{any} $\sigma_\top'\in [\tau',\frac{R}{\safety}]$, there is a an $i$ for which the Gaussian $\mathcal{G}_i$ has width $\sigma_{\top,i}$ such that $1\leq \frac{\sigma_\top'}{\sigma_{\top,i}}\leq \eta$. Thus for the Gaussian $\mathcal{G}'=\Normal(0, \sigma_\bot'^2 I_\bot + \sigma_\top'^2 I_\top)$ we invoke Lemma~\ref{lem:Pinsker} to conclude that the difference in the probability of $f(x)>z+\ferr'$ between $x\from \mathcal{G}_i$ and $x\from \mathcal{G}'$ is at most $\sqrt{\frac{n}{2}\log\eta}=\frac{\perr}{4}$.

Thus for \emph{any} $\sigma_\top'\in [\tau',\frac{R}{\safety}]$, letting $\mathcal{G}'=\Normal(0, \sigma_\bot'^2 I_\bot + \sigma_\top'^2 I_\top)$, the probability $\Pr[f(x)>z+\ferr':x\from\mathcal{G}']$ is at least $\frac{22}{32}\perr$. (Namely, given that the algorithm did not halt in Step 1b, we have a guarantee that holds over an \emph{exponentially} wide range of widths $\sigma_\bot'$, despite only taking a polynomial ($k+1$) number of iterations to check, and a polynomial ($S$) number of samples from $f$ per iteration.)

In order to apply Lemma~\ref{lem:MarkovGadget}, we need the tiny variant of the above claim, where instead of bounding the probability that $f(x)>z+\ferr'$, we instead need a bound on the probability that $f(x)-z \in (\ferr',2B)$. However, $B$ was chosen to be a truly huge number, such that we have the global guarantee of Definition~\ref{def:boundedness} that for all $x$ within distance $10nR$ of the origin, $|f(x)|\leq B$. We consider any case where our algorithm evaluates the function outside this ball to be a \emph{failure} of the algorithm. Because by the condition of Step 2e of Algorithm~\ref{alg:Ellipsoid}, the ellipsoid under consideration in Algorithm~\ref{alg:SingleCut} has semi-principal axes of length at most $3nR$, and its center lies within $R$ of the (original) origin, each point in the ellipsoid has distance at most $(3n+1)R$ from the origin. Further, by construction, each Gaussian has standard deviation in the $\bot$ dimensions bounded by the size of the ellipsoid over $s$, and standard deviation in the $\top$ dimensions at most $\frac{R}{s}$, where $s$ is chosen so that, over the entire course of the algorithm, no Gaussian sample will ever be more than $s$ standard deviations from its mean, except with negligible probability. Thus, except with negligible probability, all samples from the algorithm are in the region within $10nR$ of the origin (as a very loose bound), and thus have function value $|f(x)|\leq B$. Having analyzed this negligible probability of failure, we assume for the rest of this proof that all function evaluations have magnitude less than $B$, and condition all probabilities on the assumption that no failure has occurred.

Thus for any $\sigma_\top'\in [\tau',\frac{R}{\safety}]$ we have $\Pr[f(x)-z \in (\ferr',2B):x\from\mathcal{G}']\geq\frac{22}{32}\perr$, and
we may now invoke Lemma~\ref{lem:MarkovGadget}, in order to describe the properties of the function $g(\mu_\bot',\sigma_\top)$ at the center of Step 3. Using the parameters defined in Definition~\ref{def:Parameters}, the bound $E$ of Lemma~\ref{lem:MarkovGadget} is found to equal $\frac{1}{2}\perr$. Lemma~\ref{lem:MarkovGadget} thus yields (with distribution $D$ as defined in the lemma, and as used in Step 3 of Algorithm~\ref{alg:SingleCut}) that $$\Pr[g(\mu_\bot', \sigma_\top) > \frac{1}{4}\perr : (\mu_\bot', \sigma_\top) \from D] \ge \frac{\perr}{4(1+2\sqrt{2n}\log\frac{2B}{\epsilon'})}$$

In Step 3, we estimate each of the three terms of $g$ via Proposition~\ref{prop:Sampling}, and take enough samples to ensure that, except with negligible probability, our estimate of $g$ is accurate to within $\frac{1}{32}\perr$. Thus for each sampled $\mu_\bot', \sigma_\top$ such that $g(\mu_\bot', \sigma_\top) > \frac{1}{4}\perr$, our estimate of $g(\mu_\bot', \sigma_\top)$ will be at least $\frac{7}{32}\perr$, and thus the condition in Step 3 will succeed.

Each iteration of Step 3 of the algorithm thus succeeds with probability at least $\frac{\perr}{4(1+2\sqrt{2n}\log\frac{2B}{\epsilon'})}$, and Chernoff bounds imply that, except with failure probability $F$, Step 3 will successfully terminate in $\poly(|\log F|,n,1/\perr,\log\frac{2B}{\ferr'})$ many iterations.

Given that our observed value of $g(\mu_\bot', \sigma_\top)$ when Step 3 terminates is at least $\frac{7}{32}\perr$ and our estimates are accurate to within $\frac{1}{32}\perr$, the true value of $g(\mu_\bot', \sigma_\top)$ must be at least $\frac{6}{32}\perr$.

We now show that the gradient $\vec{d}_\bot = \Grad_\bot[L_z(x) : x \from \Normal(\mu_\bot', \sigma_\bot^2 I_\bot + \sigma_\top^2 I_\top)]$ estimated in Step 4 has positive component in the direction away from the global optimum, which will enable us to make a cut. Lemma~\ref{lem:CutLB} shows that the component of the gradient at location $\mu_\bot'$ in the direction away from the global optimum is at least \[g(\mu_\bot', \sigma_\top)-\left.\Diff{\alpha} \Exp\left[L_z(x) : x \from \Normal(\mu_\bot +\alpha \mu_\top, \sigma_\bot^2 I_\bot + \sigma_\top^2 I_\top) \right]\right|_{\alpha = 1}\]

From the above bound $g(\mu_\bot', \sigma_\top)\geq\frac{6}{32}\perr$ and the bound of Lemma~\ref{lem:MovingInThinDimension} the term we subtract from $g$ is at most $\frac{\perr}{16}$, we conclude that the component of the gradient in the direction away from the global optimum is at least $\frac{1}{8}\perr$.

We thus estimate the gradient in Step 4 by estimating the derivative in each of the dimensions in $\bot$ to within error $\frac{1}{16n}\perr$, which guarantees that the total gradient vector has component in the direction away from the global optimum at least $\frac{1}{16}\perr$. (As an implementation detail, we note that Proposition~\ref{prop:Sampling} computes a version of the derivative scaled by $\sigma_\bot$, to any desired accuracy $\kappa$, using samples scaling polynomially with $\frac{1}{\kappa}$. Thus for the unscaled derivative to have accuracy $\frac{1}{16n}\perr$, the scaled derivative must have accuracy $\frac{1}{16n}\perr\cdot\sigma_\bot$, which requires time polynomial to the inverse of this quantity. From the definition of $\sigma_\bot$ in Definition~\ref{def:Parameters}, its inverse is polynomial in the overall parameters of the algorithm.)

Thus our estimate of the gradient, normalized to a unit vector $\hat{\vec{d}}_\bot$, defines a halfspace $\{x:(x-\mu_\bot')\cdot \hat{\vec{d}}_\bot \leq 0\}$ that contains the global optimum. Further, since $\mu_\bot'$ came from a sample $\mu_\bot'\from\Normal(\vec{0}, (\sigma_\bot'^2 - \sigma_\bot^2)I_\bot)$, where $\sigma_\bot'=\frac{1}{3n\safety}$ (in Definition~\ref{def:Parameters}), and $\safety$ was chosen so that no Gaussian sample will ever be more than $s$ standard deviations from its mean, we have that $||\mu_\bot'||_2\leq\frac{1}{3n}$. Thus $\mu_\bot'\cdot\hat{\vec{d}}_\bot\leq \frac{1}{3n}$, from which we conclude that the global optimum is contained in the halfspace $\{x:x\cdot \hat{\vec{d}}_\bot\leq \frac{1}{3n}\}$, as desired.

\end{proof}

\bibliographystyle{plain}
\bibliography{STOC2016}
\newpage
\appendix
\section{The Scope of Star-Convex Functions}\label{sec:examples}
In this section we introduce some basic building blocks to construct large families of star-convex functions, with the aim of demonstrating the richness of the definition. More technical examples of star-convex functions are introduced in the rest of the paper to illustrate particular points (including Examples~\ref{ex:deterministic}, \ref{ex:rationals}, and Example~\ref{ex:output} below).

\medskip\noindent{\bf Basic Star-Convex Functions:}
\begin{itemize}
  \item[1.] Any convex function is star-convex.
  \item[2.] Even in 1 dimension, star-convex functions can be non-convex: \[f(x)=\begin{cases}
    |x| & \text{ if } |x|<1\\
    |2x| & \text{ otherwise}
  \end{cases}\]
  Or, for a Lipschitz twice-differentiable example, $f(x)=|x|(1-e^{-|x|})$ (from Nesterov and Polyak~\cite{Nesterov:2006}).
  \item[3.] In $n>1$ dimensions, take an \emph{arbitrary} positive function $g(x)$ on the unit sphere, and extend it to the origin in a way that is star-convex on each ray through the origin, for example, extending $g$ linearly to define \[f(x)=||x||_2\cdot g\left(\frac{x}{||x||_2}\right)\]
\end{itemize}

\medskip\noindent{\bf Ways to Combine Star-Convex Functions:}
Given any star-convex functions $f,g$ that have star centers at $f(0)=g(0)=0$, we can combine them to generate a new star convex function in the following ways:
\begin{itemize}
  \item[4.] A star-convex function can be shifted in the $x$ and $f$ senses so that its global optimum is at an arbitrary $x$ and $f$ value; in general, for any affine transformation $A$ and real number $c$, the function $h(x)=f(A(x))+c$ is star-convex.
  \item[5.] For any positive power $p\geq 1$, the function $h(x)=f(x)^p$ is star-convex.
  \item[6.] The sum $h(x)=f(x)+g(x)$ is star-convex.
  \item[7.] The product $f(x)\cdot g(x)$ is star-convex.
  \item[8.] The \emph{power mean}: for real number $p$, the function $h(x)=\left(\frac{f(x)^p+g(x)^p}{2}\right)^{1/p}$ is star-convex, defining powers via limits as appropriate. (The case $p=0$ corresponds to $h(x)=\sqrt{f(x)g(x)}$.)
\end{itemize}

\medskip\noindent{\bf Practical Examples of Star-Convex Function Classes:}

We combine the basic examples above to yield the following general classes of star-convex functions.

\begin{itemize}
  \item[9.] Sums of squares of monomials: the square of any monomial $(x^i y^j z^k\cdots)^2$, with $i,j,k,\ldots\in\mathbb{Z}^+$, will be star-convex about the origin, and thus sums of such terms will be star-convex despite typically being non-convex, for example $h(x,y)=x^2y^2+x^2+y^2$. Sums-of-squares arise in many different contexts.
  \item[10.] More generally, any polynomial of $|x|,|y|,|z|,\ldots$ with nonnegative coefficients is star-convex.
  \item[11.] In the standard machine learning setting where one is trying to learn a parameter vector $\theta$: for a given hypothesis $\hat{\theta}$, and each training example $X_i$, the hypothesis gives an error $L_{X_i}(\theta,\hat{\theta})$ that is (typically) a convex function of its second argument, $\hat{\theta}$. Averaging these convex losses over all $i$, via the power mean of exponent $p$ leads to a star-convex loss, as a function of $\hat{\theta}$, with minimum value 0 at $\hat{\theta}=\theta$. The paper by one of the authors~\cite{Valiant:2014} discusses a class of such star-convex loss functions which have the form \[h_{\theta,X}(\hat{\theta})=\left(\sum_i |(\hat{\theta}-\theta)\cdot X_i|^p\right)^{1/p}\]
Note that, for $p \ge 1$, the loss function is convex.
\end{itemize}

\section{Example of a Pathological Star-Convex Function}\label{ap:examples}

See the discussion before Definition~\ref{def:problem}.

\begin{example}\label{ex:output}
We present a class of minimization problems where we argue that no algorithm will be able to return a (rational) point with objective function value near the global minimum, and instead, in order to be effective, an algorithm must return a region on which ``the function value is low with high probability".

  For a given $\perr\in(0,1)$ we define a class of minimization problems in the unit square, parameterized by three uniformly random real numbers $X,Y,\theta\leftarrow [0,1]$. The star-convex function $f(x,y)$ will have unique global minimum at $f(X,Y)=0$, and is defined as \[f(x,y)=|y-Y|+\begin{cases}|x-X| & \mbox{ if } y=Y \mbox { or }frac(angle(X-x,Y-Y)\cdot 10^{10^{100}}+\theta)<\perr\\ 0 & \mbox{ otherwise}\end{cases},\] where  $frac(x)$ denotes the fractional part of $x$, and $angle(x,y)$ denotes the angle of location $(x,y)$ about the origin.

We first note that the ``if $y=Y$" condition is only possible if $Y$ is rational (since the input $y$ must be rational), and hence this condition occurs with probability 0; this condition is part of the definition of $f$ only to avoid making the entire line $y=Y$ a global minimum of $f$. We thus ignore the $y=Y$ possibility in what follows.

Fixing $x,y,X,Y$, the probability with respect to the random choice of $\theta$ that the evaluation of $f(x,y)$ falls into case 1 equals $\perr$. Since $\theta$ only affects $f(x,y)$ in case 1, if $1/\perr$ is significantly larger than the runtime of the algorithm, then with high probability the algorithm will never observe an evaluation via case 1.

Given that the algorithm never observes an evaluation via case 1, and the only dependence of $f$ on $X$ is via case 1, the algorithm will have no information about $X$, with high probability.

For any $x$ that is $\ferr$-far from $X$ (for some $\ferr$), the function $f$ will be at least $\ferr$, with probability at least $\perr$ over choices of $\theta$. Further, since regions of high function value appear every $10^{-10^{100}}$ radians about the origin, for any ball of radius $10^{-10^{99}}$ centered at $x$ and \emph{any} $\theta$, the the function $f$ will take value at least $\ferr$ with probability close to $\perr$ on a random point in the ball.

Thus, any optimization algorithm that runs in time $o(\frac{1}{\perr})$ and outputs a set that can be decomposed into $10^{-10^{99}}$-radius balls can expect the function value on its output set to be low only with $1-1/\perr$ probability. In short, we cannot expect any reasonable algorithm to return a set on which the function value is always near-optimal, or even near-optimal with probability 1. The best we can hope for is a polynomial relation between the error probability and the runtime, as we achieve in Theorem~\ref{thm:main}.

Letting the optimization algorithm specify output regions of double-exponentially small size does not help, since the set of angles for which case 1 applies could be replaced by the construction by Rudin~\cite{Rudin:1983} of a measurable set on $[0,1]$, whose intersection with \emph{any} subinterval $0<a<b<1$ has measure strictly between 0 and $ (b-a)\cdot \perr$
Furthermore, letting the optimization algorithm return rational points, instead of full-dimension sets, still does not help, since we could modify $f$ in the style of Example~\ref{ex:rationals} so that its values on all the rationals with $x$ coordinate more than $10^{-100}$ far from $X$ are high.
\end{example}

\section{Details for Adapting the Ellipsoid Algorithm}\label{ap:ellipsoid}

In this section we fill in many standard adaptations of the ellipsoid algorithm omitted from the body of Section~\ref{sec:ellipsoid}, supporting the analysis of Algorithm~\ref{alg:Ellipsoid}.

In analogy with the standard ellipsoid algorithm, the following lemma provides a means to construct a new ellipsoid from the current ellipsoid and a cutting plane.
Such scenarios can always be affinely transformed into cutting an origin-centered unit ball along a basis direction.
\begin{lemma}
\label{lem:NewEllipsoid}
Consider the unit ball $\mathbb{B}_{n-1}$ in $n$-dimensional space centered at the origin, and the half-space $\mathcal{H} = \{x \in \Real^n : x_1 \ge -\frac{1}{3n}\}$, where $x_i$ denotes the coordinate of $x$ in the $i^\text{th}$ dimension.
Then the ellipsoid $\mathcal{E}$ contains the intersection $\mathbb{B}_{n-1} \cap \mathcal{H}$, where
$$ \mathcal{E} = \left\{ x \in \Real^n : \frac{1}{(1-\frac{2}{3(n+1)})^2} \left(x_1 - \frac{2}{3(n+1)}\right)^2 + \frac{n^2-1}{n^2}\sum_{i \neq 1} x_i^2 \le 1 \right\} $$
\end{lemma}

\begin{proof}
Consider an arbitrary point $x \in \mathbb{B}_{n-1}$, then we have $\sum_{i \neq 1} x_i^2 \le 1 - x_1^2$.
Thus, it suffices to show that the function
$$ f(x_1) = \frac{1}{(1-\frac{2}{3(n+1)})^2} \left(x_1 - \frac{2}{3(n+1)}\right)^2 + \frac{n^2-1}{n^2}(1 - x_1^2) $$
satisfies $f(x_1) \le 1$ for all $x_1 \in [-\frac{1}{3n}, 1]$.

The proof is simple.
The two roots of $f(x_1) = 1$ are $\frac{-(5n+1)}{12n^2+5n+1}$ and $1$.
Also, $f'(1) > 0$, and since $f$ is quadratic in $x_1$, it must be the case that $f(x_1) \le 1$ for all $x_1 \in \left[\frac{-(5n+1)}{12n^2+5n+1}, 1\right]$.
Note that since, $\frac{1}{3n} \le \frac{5n+1}{12n^2+5n+1}$ for all $n \ge 1$, we have $f(x_1) \le 1$ for all $x \in [-\frac{1}{3n}, 1]$.
\end{proof}

We also want to upper bound the volume of the new ellipsoid relative to the original one, in order to show a fixed ratio of volume decrease at each round.
\begin{lemma}
\label{lem:Volume}
The volume of the ellipsoid constructed in Lemma \ref{lem:NewEllipsoid} is less than $e^{-\frac{1}{6(n+1)}}$ times that of the unit ball.
\end{lemma}

\begin{proof}
The ratio between the volumes is
$$ \left(1-\frac{2}{3(n+1)}\right)\left(\frac{n^2}{n^2-1}\right)^{\frac{n-1}{2}} < e^{-\frac{2}{3(n+1)}}e^{\frac{n-1}{2(n^2-1)}} = e^{-\frac{1}{6(n+1)}} $$
\end{proof}

The next two lemmas demonstrate the straightforward constructions required by Step 2e in Algorithm~\ref{alg:Ellipsoid} to keep all semi-principal axes of the feasible ellipsoid bounded by $3nR$, and keep the ellipsoid centered within distance $R$ of the origin.

\begin{lemma}\label{lem:resizing}
  Let $\mathcal{E}$ be an ellipsoid in $\mathbb{R}^n$ consisting of the points $x=(x_1,\ldots,x_n)$ such that $\sum_{i=1}^n (\frac{x_i-c_i}{a_i})^2\le 1$, namely with center $c$ and semi-principal axis lengths specified by the vector $a$. If the first $j$ semi-principal axes have lengths $a_i\geq 3n$ (for $i\leq j$), then replacing these axes with $n$, multiplying the remaining axes ($j+1,\ldots,n$) by $(n+1)/n$, and replacing the first $j$ elements of the center $c$ with 0 yields an ellipsoid $\mathcal{E}'$ defined by \begin{equation}\label{eq:ellipsoid-E}\sum_{i=1}^j \left(\frac{x_i}{n}\right)^2+\sum_{i=j+1}^n \left(\frac{n}{n+1}\cdot\frac{x_i-c_i}{a_i}\right)^2\le 1\end{equation}
  that has smaller volume than $\mathcal{E}$ and contains the entire intersection of $\mathcal{E}$ with the unit ball.
\end{lemma}
\begin{proof}
  Let $\mathcal{E}_{>j}$ be the $(n-j)$-dimensional ellipsoid, defined as $\mathcal{E}$ restricted to dimensions $>j$, namely, points $x_{>j}=(x_{j+1},\ldots,x_n)$ satisfying $\sum_{i=j+1}^n (\frac{x_i-c_i}{a_i})^2\le 1$, and let $\mathbb{B}_j$ be the $j$-dimensional unit ball centered at the origin. Then the intersection of $\mathcal{E}$ with the $n$-dimensional unit ball is contained in the cartesian product $\mathbb{B}_j\times \mathcal{E}_{>j}$. We now show that $\mathcal{E}'$ contains this cartesian product, by showing that points in this cartesian product satisfy Equation~\ref{eq:ellipsoid-E}.

  Consider a point $x\in \mathbb{B}_j\times \mathcal{E}_{>j}$. The first $j$ coordinates of $x$, corresponding to $\mathbb{B}_j$, satisfy $\sum_{i=1}^j x_i^2\leq 1$, and thus contribute to Equation~\ref{eq:ellipsoid-E} at most $\sum_{i=1}^j (\frac{x_i}{n})^2\leq \frac{1}{n^2}$. The remaining coordinates of $x$, corresponding to $\mathcal{E}_{>j}$, satisfy $\sum_{i=j+1}^n (\frac{x_i-c_i}{a_i})^2\le 1$ and thus contribute to Equation~\ref{eq:ellipsoid-E} at most $\sum_{i=j+1}^n (\frac{n}{n+1}\frac{x_i-c_i}{a_i})^2\le (\frac{n}{n+1})^2$. Adding these two bounds yields that, for points $x\in \mathbb{B}_j\times \mathcal{E}_{>j}$, the left hand side of Equation~\ref{eq:ellipsoid-E} is at most $\frac{1}{n^2}+(\frac{n}{n+1})^2\leq 1$, as desired, proving the lemma.
\end{proof}

\begin{lemma}\label{lem:recentering}
  The intersection of a ball $\mathcal{B}$ and an ellipsoid $\mathcal{E}$ is contained within the intersection of the ellipsoid and a version of the ball translated to be centered at the nearest point of $\mathcal{E}$ to the center of $\mathcal{B}$.
\end{lemma}

Lemma~\ref{lem:recentering}, interpreted under an affine transformation, and combined with Lemma~\ref{lem:resizing}, shows that we can run the ellipsoid method with every ellipsoid guaranteed to have axes bounded by $3nR$, and always centered inside the ball of radius $R$.

The following two lemmas give important invariants that the algorithm maintains throughout the iterations.

\begin{lemma}
The global minimum is contained in the ellipsoid $\mathcal{E}_i$ for all $i$.
\end{lemma}

\begin{proof}
We prove by induction.
The base case is trivial.

For the inductive case, assume the lemma is true for the ellipsoid $\mathcal{E}_i$.
Apply an affine transformation such that the $\mathcal{E}_i$ is a unit ball.
By assumption, the cut $\vec{d}_\bot$ satisfies the property that, taking into account the affine transformation, the half-space $\mathcal{H} = \{x \in \Real^n : x\cdot\hat{\vec{d}}_\bot \ge -\frac{1}{3n}\}$ contains the global minimum.
Since the global minimum is also contained in $\mathcal{E}_i$ by assumption, it must be contained in $\mathcal{E}_{i+1}$ by Lemma \ref{lem:NewEllipsoid}.
\end{proof}

\begin{lemma}
\label{lem:AxesLB}
No semi-principal axis of any ellipsoid $\mathcal{E}_i$ is ever less than $\left(\frac{1+\frac{1}{3n}}{2}\right)\tau$.
\end{lemma}

\begin{proof}
We prove by induction.
The base case is again trivial.

For the inductive case, assume the lemma is true for the ellipsoid $\mathcal{E}_i$.
For simplicity, when we refer to thin (less than $\tau$) and non-thin (at least $\tau$) directions, they are always with respect to the lengths of the axes of $\mathcal{E}_i$.
Observe that, by construction, all the thin directions in the axes of $\mathcal{E}_i$ are also axes directions of $\mathcal{E}_{i+1}$.
Also, by Step 2(c) in Algorithm \ref{alg:Ellipsoid}, $\mathcal{E}_{i+1}$ always contain the center of $\mathcal{E}_i$.
Therefore, the semi-principal axes of $\mathcal{E}_{i+1}$ in the thin directions must be longer than that of $\mathcal{E}_i$, which is at least the quantity in the lemma statement by the induction hypothesis.
Now consider an arbitrary axis $\vec{v}$ of $\mathcal{E}_{i+1}$ that is not in any of the thin directions.
The direction of this axis is a linear combination of the non-thin directions in the axes of $\mathcal{E}_i$.
Therefore, the diameter of $\mathcal{E}_i$ in the direction of $\vec{v}$ is at least $2\tau$.
By construction of the cut, $2||\vec{v}||_2 \ge (1+\frac{1}{3n})\tau$, and so
$||\vec{v}||_2 \ge \left(\frac{1+\frac{1}{3n}}{2}\right)\tau$,
completing the proof of the lemma.
\end{proof}

We restate and prove Lemma~\ref{lem:EllipsoidHalt}, showing that Algorithm~\ref{alg:Ellipsoid} halts after the designated number of iterations.
\begin{lemma*}[Lemma~\ref{lem:EllipsoidHalt}]
Algorithm~\ref{alg:Ellipsoid} halts within $m+1$ iterations either through Step 2a or Step 2c.
\end{lemma*}

\begin{proof}
By Lemma~\ref{lem:Volume} and induction, after $m$ iterations, the volume of ellipsoid $\mathcal{E}_m$ is at most $\left(\left(\frac{1+\frac{1}{3n}}{2}\right)\tau\right)^{n-1}\tau$ times the volume of the input radius $R$ ball.
Since, by Lemma~\ref{lem:AxesLB}, all its semi-principal axes have length at least $\left(\frac{1+\frac{1}{3n}}{2}\right)\tau$, all the axes must also be of length at most $\tau$.

Therefore, the algorithm will definitely halt in Step 2a of iteration $m+1$, if it has not already done so in the previous iterations.
\end{proof}

\end{document}